\RequirePackage{fix-cm}
\documentclass[natbib]{svjour3}       

\smartqed  
 \journalname{$$}
\usepackage[latin9]{inputenc}
\usepackage{amsthm}
\usepackage{amsmath}
\usepackage{amssymb}
\usepackage{esint}

\pdfpageheight\paperheight
\pdfpagewidth\paperwidth

\providecommand{\tabularnewline}{\\}

\theoremstyle{plain}
\newtheorem{thm}{\protect\theoremname}
  \theoremstyle{plain}
  \newtheorem{prop}{\protect\propositionname}
  \theoremstyle{remark}
  \newtheorem{rem}{\protect\remarkname}
  \theoremstyle{plain}
  \newtheorem{lem}{\protect\lemmaname}
  \theoremstyle{plain}
  \newtheorem{cor}{\protect\corollaryname}

\usepackage{graphicx}
\usepackage[heavycircles,only, ovee,varovee,owedge,varowedge]{stmaryrd}
\usepackage[colorlinks,linkcolor=blue,citecolor=blue,urlcolor=blue]{hyperref}

\theoremstyle{definition}
\newtheorem{defn}[thm]{Definition}

\def\Fr {\mathfrak{F}}
\def\E{\mathbb E}
\def\P{{\mathbb P}}
\def\R{{\mathbb R}}
\def\N{{\mathbb N}}

\makeatother

  \providecommand{\lemmaname}{Lemma}
  \providecommand{\propositionname}{Proposition}
  \providecommand{\remarkname}{Remark}
\providecommand{\corollaryname}{Corollary}
\providecommand{\theoremname}{Theorem}

\begin{document}

\title{CRPS M-estimation for max-stable models
\thanks{R.Y. was partially funded by UM Rackham Merit Fellowship and NSF-AGEP grant DMS 1106695. S.S was partially funded by NSF grant DMS 1106695}  
}


\author{Robert Yuen      
 \and
        Stilian Stoev 
}


\institute{R.A. Yuen \at University of Michigan \\
              439 West Hall, 1085 South University Ave. Ann Arbor, MI 48109-1107 \\
              Tel.: 734.763.3519\\
              Fax: 734.763.4676\\
              \email{bobyuen@umich.edu}           
           \and
           S. Stoev \at University of Michigan
}

\date{June 5th, 2013}

\maketitle

\begin{abstract}
Max-stable random fields provide canonical models for the dependence of multivariate extremes. Inference with such models has been 
challenging due to the lack of tractable likelihoods. In contrast, the finite dimensional cumulative distribution functions (CDFs) 
are often readily available and natural to work with.  Motivated by this fact, in this work we develop an M-estimation framework for 
max-stable models based on the {\em continuous ranked probability  score} (CRPS) of multivariate CDFs. We start by establishing 
conditions for the consistency and asymptotic normality of the CRPS-based estimators in a general context. We then implement them 
in the max-stable setting and provide readily computable expressions for their asymptotic covariance matrices. The resulting point 
and  asymptotic confidence interval estimates are illustrated over popular simulated models. They enjoy accurate coverages and offer 
an alternative to composite likelihood based methods. 
\end{abstract}

\section{Introduction}

Max-stable processes are a canonical class of statistical models for
multivariate extremes. They appear in a variety of applications ranging from 
insurance and finance \citep{embrechts:kluppelberg:mikosch:1997,finkenstadt:rootzen:2004} to spatial extremes such as precipitation
 \citep{DavisonBlanchetAoAP2012,DavisonStatSci2012}
and extreme temperature. Max-stable processes
are exactly the class of non-degenerate stochastic processes that
arise from limits of independent component-wise maxima. This fact provides a theoretical
justification for their use as models of multivariate extremes. However, many
useful max-stable models suffer from intractable likelihoods, thus
prohibiting standard maximum likelihood and Bayesian inference. This
has motivated development of  maximum \emph{composite} likelihood estimators (MCLE)
for max-stable models \citep{PadoanRibatetSisson2010} as
well as certain approximate Bayesian approaches \citep{ReichShaby2012,ErhardtSmith2011}.

In contrast to their likelihoods, the cumulative distribution functions
(CDFs) for many max-stable models are available in closed form, or
they are tractable enough to approximate within arbitrary
precision. This motivates statistical inference based on the minimum
distance method \citep{Wolfowitz1957,ParrSchucany1980}.
In this paper, we propose an M-estimator for parametric max-stable
models based on minimizing distances of the type
\begin{equation}
\int_{\mathbb{R}^{d}}\left(F_{\theta}\left(x\right)-F_{n}\left(x\right)\right)^{2}\mu\left(dx\right).\label{eq:CramerVonMisesdistance}
\end{equation}
where $F_{\theta}$ is a $d$-dimensional CDF of a parametric model,
$F_{n}$ is a corresponding empirical CDF and $\mu$ is a tuning measure
that emphasizes various regions of the sample space $\mathbb{R}^{d}.$
Using elementary manipulations it can be shown that minimizing distances of the type \eqref{eq:CramerVonMisesdistance}
is equivalent to minimizing the \emph{continuous ranked probability
score} (CRPS).
\begin{defn}
\label{def:(CRPS-M-estimator)}(CRPS M-estimator) Let $\mu$ be a
measure that can be tuned to emphasize regions of a sample space $\mathbb{R}^{d}$.
Define the CRPS functional
\begin{equation}
\mathcal{E}_{\theta}\left(x\right)=\int_{\mathbb{R}^{d}}\left(F_{\theta}\left(y\right)-\mathbf{1}_{\left\{ x\le y\right\}}\right)^{2}\mu\left(dy\right)\label{eq:CRPS}
\end{equation}
Then for independent random vectors $\left\{ X^{\left(i\right)}\right\} _{i=1}^{n}$
with common distribution function $F_{\theta_{0}}$ we define the
following CRPS M-estimator for $\theta_{0}$.

\begin{equation}
\hat{\theta}_{n}=\underset{\theta\in\Theta}{\text{argmin}}\sum_{i=1}^{n}\mathcal{E}_{\theta}\left(X^{\left(i\right)}\right).\label{eq:CRPScriterion}
\end{equation}
For simplicity, we shall assume throughout that the parameter space
$\Theta$ is a compact subset of $\mathbb{R}^{p}$, for some integer
$p$.
\end{defn}
The remainder of this paper is organized as follows. In Section \ref{sec:Extreme-values-and-nax-stability}
we review some essential multivariate extreme value theory and provide
definitions and constructions of max-stable models. In Section \ref{sec:CRPS-M-estimation}
we establish regularity conditions for consistency and asymptotic
normality of the CRPS M-estimator and provide general formulae for
calculating its asymptotic covariance matrix. In Section \ref{sec:CRPS-for-max-stable}
we specialize these calculations to the max-stable setting. In Section
\ref{sec:Simulation} we conduct a simulation study to evaluate the
proposed estimator for popular max-stable models. 

\section{\label{sec:Extreme-values-and-nax-stability}Extreme values and max-stability}

Let  $Y^{(i)} = \{Y^{(i)}_t\}_{t\in T},\ i=1,2,\cdots$ be independent and identically distributed measurements of certain environmental or physical phenomena. For example, the $Y_t^{(i)}$s may  model wave-height, temperature, precipitation, or pollutant concentration levels at a site $t$ in a spatial  region $T\subset \R^2$. 
If one is interested in extremes, it is natural to consider the asymptotic behavior  of the point-wise maxima. Suppose that, for some $a_n(t)>0$ and $b_n(t)\in \R$, we have 

\begin{equation}\label{eq:extremevalueprocess}   
  {\Big\{} \frac{1}{a_n(t)} \max_{i=1,\cdots,n} Y_t^{(i)}-b_{n}(t) {\Big\}}_{t\in T} \stackrel{d}{\longrightarrow}  \{X_{t}\}_{t\in T},\ \mbox{ as }n\to\infty,    
\end{equation} 

 for some non-trivial limit process $X$, where $\stackrel{d}{\to}$ denotes convergence of the finite-dimensional  distributions. The class of {\em extreme value} 
 processes $X=\{X_t\}_{t\in T}$ arising in the limit describe the statistical dependence of `worst case scenaria' and are therefore natural models of multivariate extremes.
The limit $X$ in \eqref{eq:extremevalueprocess} is necessarily a {\em max-stable process} in the sense that for all $n$, there exist  $c_n(t)>0$ and $d_n(t)\in\R$, such that 
$$
 {\Big\{} \frac{1}{c_n(t)} \max_{i=1,\cdots,n} X_t^{(i)} - d_n(t) {\Big\}}_{t\in T} \stackrel{d}{=} \{X_t\}_{t\in T}, 
$$ 
where $\{X^{(i)}_t\}_{t\in T}$ are independent copies of $X$ and where $\stackrel{d}{=}$ means equality of finite-dimensional 
distributions \citep[Ch.5 of ][]{Resnick1987}. Due to the classic results of Fisher-Tippett and Gnedenko, the marginals of  $X$ are 
necessarily extreme value distributions (Fr\'echet, reversed Weibul or Gumbel). 
They can be described  in a unified way through the generalized extreme value distribution (GEV): 
\begin{equation} 
   G_{\xi,\mu,\sigma}(x):=\exp{\Big\{} - (1+\xi (x-\mu)/\sigma)_{+}^{-1/\xi}{\Big\}} ,\ \sigma>0, \label{eq:GEVcdf-1} 
\end{equation} 
where $x_{+}=\max\{ x,0\}$, and where $\mu, \sigma$ and $\xi$ are known as the location, scale and shape parameters. 
The cases $\xi>0, \xi<0,$ and $\xi\to0$ correspond to Fr\'{e}chet, reverse Weibull, and Gumbel,  respectively 
(see, e.g. Ch.3 and 6.3 in \citealp{embrechts:kluppelberg:mikosch:1997} for more details).

The dependence structure of the limit extreme value process $X$ rather than its marginals is of utmost interest in practice. Arguably, the type 
of the marginals is unrelated to the dependence structure of $X$ and as it is customarily done, we shall assume that the limit process $X$ has been transformed to standard $1$-Fr\'echet marginals. That is, 
\begin{equation}\label{e:Xt-cdf} 
 \P( X_t \le x)  = G_{1,1,\sigma_t}(x) =  e^{-\sigma_t/x},\ x>0, 
\end{equation} for some scale $\sigma_t>0$ \citep[Ch.5 of ][]{Resnick1987}. 
 
\subsection{Representations of max-stable processes}

Let $X = \{X_t\}_{t\in T}$ be a max-stable process with $1$-Fr\'echet marginals as in \eqref{e:Xt-cdf}.  Then, its finite-dimensional  distributions are multivariate max-stable random vectors and they have the following representation: 
\begin{equation}\label{e:spec_measure}  
  \P(X_{t_i}\le x_i,\ i=1,\cdots,d)    = \exp{\Big\{} - \int_{{\mathbb S}_+^{d-1}} {\Big(}\max_{i=1,\cdots,d} w_i/x_i {\Big)} H(dw)  {\Big\}}, 
 \end{equation} 
 where $x_i>0,\ t_i\in T,\ i=1,\cdots,d$ and where $H = H_{t_1,\cdots,t_d}$ is a finite measure on the positive unit sphere  
 $$
 {\mathbb S}_+^{d-1} = \{ w = (w_i)_{i=1}^d\, :\, w_i\ge 0,\ \sum_{i=1}^d w_i = 1\} $$ 
 known as the {\em spectral measure} of the max-stable random vector $(X_{t_i})_{i=1}^d$ \citep[see e.g. Proposition 5.11 in ][]{Resnick1987}.
The integral in the expression \eqref{e:spec_measure} is referred to as the {\em tail dependence function} of the max-stable law. We shall often use the notation: 
$$ 
V(x) \equiv V_{t_1,\cdots,t_d}(x) := -\log \P(X_{t_i}\le x_i,\ i=1,\cdots,d),   
$$ 
where $x=(x_i)_{i=1}^d\in \R_+^d$, for the tail dependence function of the max-stable random vector $(X_{t_i})_{i=1}^d$. 

It readily follows from \eqref{e:spec_measure}  that for all $a_i\ge 0, i=1,\cdots,d$, the max-linear combination
$$
 \xi:= \max_{i=1,\cdots,d} a_i X_{t_i}
$$
is $1$-Fr\'echet random variable with scale $\sigma_\xi = \int_{{\mathbb S}_+^{d-1}}( \max_{i=1,\cdots,d} a_i w_i ) H(dw)$. Conversely, a random vector $(X_{t_i})_{i=1}^d$ with the
property that all its non-negative max-linear combinations are $1$-Fr\'echet is necessarily multivariate max-stable \citep{dehaan:1978}. This invariance to max-linear
combinations is an important feature that will be used in our estimation methodology (Section \ref{sec:CRPS-for-max-stable}, below).

Some max-stable models are readily expressed in terms of their spectral measures while others via tail dependence  functions. 
These representations however are not convenient for computer simulation or in the case of random processes, where one needs a handle on all finite-dimensional distributions.
The most common constructive representation of max-stable process models is based on Poisson point processes  \citep{deHaan1984,schlather:2002,kabluchko:schlather:dehaan:2009}. See also \cite{stoev:taqqu:2005} for an  alternative. 

Indeed, consider a measure space $(S,\mathcal{S},\nu)$ and let $\Pi:=\{ (\epsilon_{i},S_{i})\} _{i\in\mathbb{N}}$ be a Poisson point process on $\mathbb{R}^{+}\times S$ with intensity measure $dx d\nu$. 

\begin{prop}\label{p:de-Haan-simple} Let $g_t \in L^1(S,{\cal S}, \nu),\ t\in T$ be a collection of non-negative integrable functions and let  
\begin{equation}  
   X_{t} := \int_{S}^\vee g_t d\Pi \equiv \max_{i\in\mathbb{N}} \epsilon_i^{-1} g_t(S_i),\ \ (t\in T).   \label{eq:maxLinRF} 
  \end{equation} 
  Then, the process $X = \{X_t\}_{t\in T}$ is max-stable with $1$-Fr\'echet marginals and  finite-dimensional distributions: 
  \begin{equation}\label{e:fdd-de-Haan}  
     \P(X_{t_i} \le x_i,\ i=1,\cdots,d) = \exp{\Big\{} - \int_{S}{\Big(} \max_{i=1,\cdots,d} g_{t_i}(s)/x_i {\Big)} \nu(ds) {\Big\}}. 
    \end{equation} \end{prop}
 The proof of this result is sketched in Appendix \ref{sub:Derivation-of-max-stableCDF}. Relation \eqref{eq:maxLinRF} is known as the {\em de Haan spectral representation} of $X$ and 
 $\{g_t\}_{t\in T} \subset L_+^1(S,{\cal S},\nu)$ as the spectral functions of the process. It can be shown that every separable in probability max-stable process 
 has such a representation (see \citealp{deHaan1984} and Proposition 3.2 in \citealp{stoev:taqqu:2005}). 
 
 The max-functional in \eqref{eq:maxLinRF} has the properties of an {\em extremal stochastic integral}. Indeed, we have
 max-linearity:
 $$
  \max_{i=1,\cdots,d} a_i X_{t_i} = \int_S^\vee {\Big(}\max_{i=1,\cdots,d} a_i g_{t_i} {\Big)} d\Pi,
 $$
 for all $a_i\ge 0.$ The above max-linear combination is therefore $1$-Fr\'echet and has a scale coefficient:
 $$
 \int_S {\Big(}\max_{i=1,\cdots,d} a_i g_{t_i} {\Big)} d \nu = {\Big\|} \max_{i=1,\cdots,d} a_i g_{t_i} {\Big\|}_{L^1(\nu)}.
 $$
 One can also show that $X_t$ and $X_s$ are independent, if and only if $g_t(u)g_s(u) = 0$, for $\nu$-almost all $u\in S$.
 That is, the extremal integrals defining $X_t$ and $X_s$ are over non-overlapping sets.  This shows that 
 for max-stable process models pairwise independence implies independence. Further, $X_{t_n}$ converges in 
 probability to $X_t$ if and only if $g_{t_n}$ converges in $L^1(\nu)$ to $g_t$, as $n\to\infty$. For more details, 
 see e.g \cite{deHaan1984} and \cite{stoev:taqqu:2005}.
 
 \begin{rem} The expressions \eqref{e:spec_measure} and \eqref{e:fdd-de-Haan} may be related through a change of variables  
(Proposition 5.11 \citealp{Resnick1987}). While the spectral measure $H$ in \eqref{e:spec_measure} is unique,  a max-stable process 
has many different spectral function representations. Nevertheless, Relation \eqref{eq:maxLinRF} provides a constructive and 
intuitive representation of $X$, that can be used to build interpretable models. 
 \end{rem}

\subsection{Max-stable models} 
A great variety of max-stable models can be defined by specifying the measure space $(S,\mathcal{S},\nu)$ and an accompanying family of 
spectral functions $g_t$ or equivalently through a consistent family of spectral measures or tail dependence functions.
 We review next several  popular max-stable models and their basic features.

\medskip
$\bullet$ {\em (Multivariate logistic)} Let $X = (X_{t_i})_{i=1}^d$ have the CDF
$$
 F_X(x) = e^{-V(x)},\ \ \mbox{ where } V(x) = \sigma \times {\Big(} \sum_{i=1}^d x_{t_i}^{-1/\alpha} {\Big)}^{\alpha},
$$
for $\sigma>0$ and $\alpha \in [0,1].$ The parameter $\alpha$ controls the degree of dependence, where $\alpha = 1$ 
corresponds to independence ($V(x) = \sigma \sum_{i=1}^d x_i^{-1}$), while $\alpha=0$  to complete dependence 
($V(x) = \sigma \max_{i=1,\cdots,d} x_i^{-1}$, interpreted as a limit). 

This model is rather simple since the dependence is exchangeable but it provides a useful benchmark for the 
performance of the CRPS-based estimators since the MLE is easy to obtain in this case (see Table \ref{tab:Logistic-model-simulation} below). 
The recent works of \cite{fougeres:nolan:rootzen:2009} and \cite{fougeres:mercadier:nolan:2013} develop far-reaching 
generalizations of multivariate logistic laws by exploiting connections to sum-stable distributions.

\medskip
$\bullet$ {\em (Max-linear or spectrally discrete models)} Let $A = (a_{ij})_{d\times k}$ be a matrix with non-negative entries and let $Z_j,\ j=1,\cdots,k$ be independent standard $1$-Fr\'echet random variables. Define
\begin{equation}\label{e:max-linear}
X_i = \max_{j=1,\cdots,k} a_{ij} Z_j,\ \ i=1,\cdots,d.
\end{equation} 
The vector $X = (X_i)_{i=1}^d$ is max-stable.  It can be shown that the CDF of $X$ has the form
\eqref{e:spec_measure} were the spectral measure
 \begin{equation}\label{e:disc_spec_measure}
  H (dw) = \sum_{j=1}^k |a_{\cdot j} | \delta_{\{ a_{\cdot j}/|a_{\cdot j}|\}}( dw),
 \end{equation}
 is concentrated on the normalized column-vectors of the matrix $A$, i.e.\ on
 $a_{\cdot j}/|a_{\cdot j}| := (a_{ij}/|a_{\cdot j}|)_{i=1}^d$, where $|a_{\cdot j}| = \sum_{i=1}^d a_{ij}$, and where $\delta_{a}$ stands for the Dirac measure with unit mass at the point $a\in \mathbb{R}^d$.
 
 Conversely, any max-stable random vector with discrete spectral measure $H$ has a max-linear representation as in
 \eqref{e:max-linear}, where the columns of the matrix $A$ may be recovered from \eqref{e:disc_spec_measure}. We shall
 also call such models spectrally discrete.
 
 Since any spectral measure $H$ can be approximated arbitrarily well with a discrete one, max-linear models are dense 
 in the class of all max-stable models. As argued in \cite{einmahl:krajina:segers:2012}, max-linear distributions arise naturally
 in economics and finance, as models of extreme losses. The $Z_j$s represent independent
 shock-factors that lead to various extreme losses in a portfolio $X$ depending on the factor loadings $a_{ij}$. 
 
 Max-linear models are particularly well-suited for CRPS-based inference, since their tail dependence function has
 a simple closed form:
 \begin{equation}\label{e:V-max-linear}
  V(x) = \sum_{j=1}^k \max_{i=1,\cdots, d} a_{ij}/x_i,\ \ x=(x_i)_{i=1}^d \in \R_+^d.
 \end{equation}
 See Section \ref{sec:Simulation} below for a simple example of CRPS-based inference for max-linear models 
 and \cite{einmahl:krajina:segers:2012} for an alternative M-estimation methodology.
 
 \medskip
 $\bullet$ {\em (Moving maxima and mixed moving maxima)} Let $(S,{\cal S},\nu) \equiv (\R^k,{\cal B}_{\R^k},{\rm Leb})$
 and $g_t(s):= g(t-s), t,s\in\R^k$, for some non-negative integrable function $g\ge 0$, $\int_{\R^k} g(s) ds<\infty$.
 Then \eqref{eq:maxLinRF} yields the so-called {\em moving maxima} random field:
 $$
  X_t := \int^{\vee}_{\R^k} g(t-s) d \Pi(s) \equiv \max_{i\in \N} g(t-S_i)/\epsilon_i,\ \ (t\in \R^k).
 $$
 The choice of the kernel $g$
  as a multivariate Normal density in $\R^2$ yields the well-known {\em Smith storm model}, where the $S_i$s may be 
 interpreted as storm locations, $g$ is the spatial storm attenuation profile and $1/\epsilon_i$ its strength.
 
 More flexible models can be obtained by taking maxima of independent moving maxima, resulting in the so-called 
 {\em mixed moving maxima}:
 \begin{equation}\label{e:mmm}
  X_t = \int_{\R^k \times U}^\vee g(t-s,u) d \Pi(s,u) \equiv \max_{i\in \N} g(t-S_i,U_i)/\epsilon_i,\ \ (t\in \R^k)
 \end{equation}
 where $\Pi$ is a Poisson point process on $S = \R^k\times U$ with intensity $\nu(ds,du) = ds m(du)$, and where
 $g\ge 0$ is such that $\int_{\R^k\times U} g(s,u) ds m(du) <\infty$. Here $m(du)$ is the `mixing' measure, which 
 may be continuous or discrete, and the $U_i$s may be viewed as different types of storms.
 
 The mixed moving maxima random fields are stationary, ergodic and, in fact, mixing 
 \citep{stoev:2008,kabluchko:schlather:2010}. By \eqref{e:fdd-de-Haan}, their tail dependence functions are 
 $$
  V(x) = \int_{\R^k \times U} {\Big(} \max_{i=1,\cdots,d} g(t_i-s,u)/x_i {\Big)} ds m(du),\ \ x=(x_i)_{i=1}^d\in \R_+^d.
 $$
 
 \medskip
 $\bullet$ {\em (Spectrally Gaussian models)} By viewing $(S,{\cal S},\nu)$ as a probability space, in the case
 $\nu(S)=1$, the spectral functions $\{g_t\}_{t\in T}$ in \eqref{eq:maxLinRF} become 
 a stochastic process. By  picking $g_t = h(w_t)$ to be non-negative transformations of a Gaussian process $w_t$ on 
 $(S,{\cal S},\nu)$, one obtains interesting and tractable max-stable models whose dependence structure is governed by the covariance structure 
 of the underlying Gaussian process $\{w_t\}_{t\in T}$. The popular Smith, Schlather, and Brown-Resnick random field
 models are of this type \citep{Smith1990,schlather:2002,brown:resnick:1977,stoev:2008,kabluchko:schlather:dehaan:2009}.
 
 \medskip
  $\circ$ {\em (Schalther models)} Let $\{w_t\}_{t\in \R^k}$ be a stationary Gaussian random field with zero mean and
  let $g_t(s) := (w_t(s))_+,\ s\in S$. Then $X_t$ in \eqref{eq:maxLinRF} has the following tail dependence function
  \begin{equation}\label{e:V-schlather}
  V(x) = \E_\nu {\Big(} \max_{i=1,\cdots,d} (w_{t_i})_+/ x_i {\Big)},\ \ x = (x_i)_{i=1}^d\in \R_+^d,
  \end{equation}
  where $\E_\nu$ denotes integration with respect to the `probability' measure $\nu$.

\medskip
$\circ$ {\em (Brown-Resnick)} Let $w=\{w_t\}_{t\in \R^k}$ be a zero mean Gaussian random field with stationary
 increments.
Set $g_t(s):= e^{w_t(s) - v_t/2}$, where $v_t = \E_\nu (w_t^2)$ is the `variance' of $w_t$. The seminal paper of
 \cite{brown:resnick:1977} introduced this model with $w$ -- the standard Brownian motion and showed that, surprisingly,
the resulting max-stable process $X_t$ in \eqref{eq:maxLinRF} is stationary, even though $w$ is not.
The cornerstone work of \cite{kabluchko:schlather:dehaan:2009} showed that $\{X_t\}_{t\in \R^k}$ is stationary
for a centered Gaussian process $w$, with stationary increments . It also obtained important mixed moving
maxima representations of $X$ under further conditions on $w$. The tail dependence function of $X$ in this case is
 \begin{equation}\label{e:V-BR}
  V(x) = \E_\nu {\Big(} \max_{i=1,\cdots,d} e^{w_{t_i} - v_{t_i}/2}/ x_i {\Big)},\ \ x = (x_i)_{i=1}^d\in \R_+^d.
  \end{equation}

It can be shown that the Smith model \citep{Smith1990} is a special case of a Brown-Resnick model with a degenerate random field
$\{w_t\} \stackrel{d}{=} \{ t^\top Z\},\ t\in \R^k$, $k<d$, where $Z$ is a Normal random vector in $\R^k$. The above models 
can be deemed {\em spectrally Gaussian} since their tail dependence functions (and hence spectral measures)
are expectations of functions of Gaussian laws. One can consider other stochastic process
models for the underlying spectral functions $g_t$ and thus arrive at {\em doubly stochastic} max-stable processes.
We comment briefly on some general probabilistic properties of these models.
 
\begin{rem}
 If $\{g_t\}_{t\in \R^k}$ is a stationary process in $(S,{\cal S},\nu)$, then the max-stable process
$X=\{X_t\}_{t\in \R^k}$ is also stationary.  It is, however, non-ergodic. In particular, the Schlather models are non-ergodic.
 This is important in applications, since a single 
observation of the random field $X$ at an expanding grid, may not yield consistent parameter estimates.

\cite{kabluchko:schlather:dehaan:2009} have shown that Brown-Resnick random fields with non-stationary 
$\{w_t\}$  such that $\lim_{|t|\to\infty} (w_t - v_t/s) = -\infty,$ almost surely, 
have mixed moving maxima representations as in \eqref{e:mmm}. They are therefore mixing \citep{stoev:2008} and consistent 
statistical inference from a single realization of such max-stable random fields is possible.
\end{rem} 

\begin{rem}
The Poisson point process construction in \eqref{eq:maxLinRF} involves a maximum over an infinite number of terms.  
As a result, computer simulations of spectrally Gaussian max-stable models necessitates truncation to a finite number.
In the case of the Brown-Resnick model, the number of terms required to produce a satisfactory 
representation is prohibitively large.  Accurate simulation of Brown-Resnick processes is an
active area of study \citep{Oesting2011}.  Consequently, simulation studies for inference under
Brown-Resnick models have yet to appear.  For this reason the remaining discussion of spectrally Gaussian
max-stable models including simulation and application is restricted to the Schlather model.

\end{rem}

 \subsection{Measures of dependence in max-stable models}$\quad$
 \medskip
 
 $\bullet$ {\em (Co-variation)} For $X_t$ as in \eqref{eq:maxLinRF}, define
$$
[X_t,X_s] := \int_{S} g_t \wedge g_s  d\nu \equiv \int_S g_t d\nu + \int_S g_s d\nu - \int_S g_t\vee g_s d\nu,\ (t,s\in T).  
$$
Note that $\int_S g_t d\nu$ and $\int_S (g_t\vee g_s) d\nu$ are the scale coefficients of the Fr\'echet 
random variables $X_t$ and $X_t\vee X_s$. The co-variation $[X_t,X_s]\ge 0$ is non-negative and 
equals zero if and only if $X_t$ and $X_s$ are independent, analagous to covariance for Gaussian processes.

\medskip
$\bullet$ {\em (Extremal coefficient)} A popular summary measure of multivariate dependence in max-stable
models is the extremal coefficient. Define
\[
\vartheta\left(D\right):=-\log\mathbb{P}\left(X_{t}\le1,t\in D\right)\equiv V\left(\mathbf{1}\right).
\]
For a process $\left\{ X_{t},t\in D\right\} $ with standard 1-Fr\'{e}chet
marginals 
\[
\max_{t\in D}\frac{1}{x_{t}}\le V\left(x\right)\le\sum_{t\in D}\frac{1}{x_{t}}
\]
and thus $1\le\vartheta\left(D\right)\le d=\vert D \vert $, where $\vartheta\left(D\right)=1$
corresponds to complete dependence while $\vartheta\left(D\right)=d$
implies that $X_{t}$'s , $t\in D$ are independent. 

\medskip
It is well know that for a process with standard $1$-Fr\'echet marginals, $\vartheta(\{t,t+h\}) = 2 - [X_t,X_{t+h}]$.
In the case of the Schlather model there is an explicit formula for the bivariate
extremal coefficient in terms of the correlation function: $\vartheta(\{t,s\}) = 1+\sqrt{(1-\rho(t,s))/2}$. 
Figure \ref{fig:max-stable-realizations} displays realizations from the Schlather model
for the different  correlation functions given in Table \ref{tab:Spectrally-Gaussian-max-stable-models}.  Note that these examples are 
all (spectrally) isotropic in the sense
that the correlation $\rho\left(t,s\right)$ of the underlying Gaussian process
depends only on the distance $h=\left\Vert t-s\right\Vert $ between
locations $t$ and $s$. This however is not a requirement in general.
Figure \ref{fig:max-stable-realizations} also provides some visual
evidence of how the covariance structure and smoothness of $w$ influence the dependence
structure of the resultant max-stable random field $X$. It is
possible to parameterize the dependence structure of the max-stable
random field using a large variety of covariance functions available
for parameterizing Gaussian processes. 
 
 \renewcommand{\arraystretch}{2}

\begin{table}
\caption{\label{tab:correlation-functions}Correlation functions for Gaussian random fields. For the Mat\'{e}rn covariance function, $K_{\theta_2}$ 
is the modified Bessel function of the second kind. \label{tab:Spectrally-Gaussian-max-stable-models}}

\begin{centering}
\begin{tabular}{c|c}
& $\rho_{\theta} (t,s)$, $h=\Vert t-s \Vert $ \tabularnewline
\hline
Stable & $\exp \big[ - (h/\theta_1)^{\theta_2} \big] \quad \theta_1 >0, \theta_2 \in (0,2]$  \tabularnewline
Mat\'{e}rn  &  $ \frac{(\sqrt{2\theta_2}h/\theta_1)^{\theta_2}}{\Gamma (\theta_2) 2^{\theta_2 -1}}{ K_{\theta_2} \big( \sqrt{2\theta_2}h/\theta_1 \big)} \quad \theta_1 >0, \theta_2 > 0$ \tabularnewline
Cauchy  & $(1+(h/{\theta_1})^2)^{-\theta_2} \quad \theta_1 >0, \theta_2 > 0$\tabularnewline
\end{tabular}
\par\end{centering}

\end{table}

\renewcommand{\arraystretch}{1}

\begin{figure}

\caption{\label{fig:max-stable-realizations} Schlather max-stable model realizations using correlation functions of Table \ref{tab:correlation-functions} under varying
parameter settings.
{\bf Top:} Stable correlation function.
{\bf Middle:} Mat\'{e}rn correlation function.
{\bf Bottom:} Cauchy correlation function. Realizations were generated using the {\tt R} package {\tt SpatialExtremes} \citep{Ribatet:spatialextremes-Rpackage}.
The circles indicate locations of ``observation staions'' in the simulation study of Section \ref{sec:Simulation}.}

\begin{centering}
\includegraphics[width=2.25in]{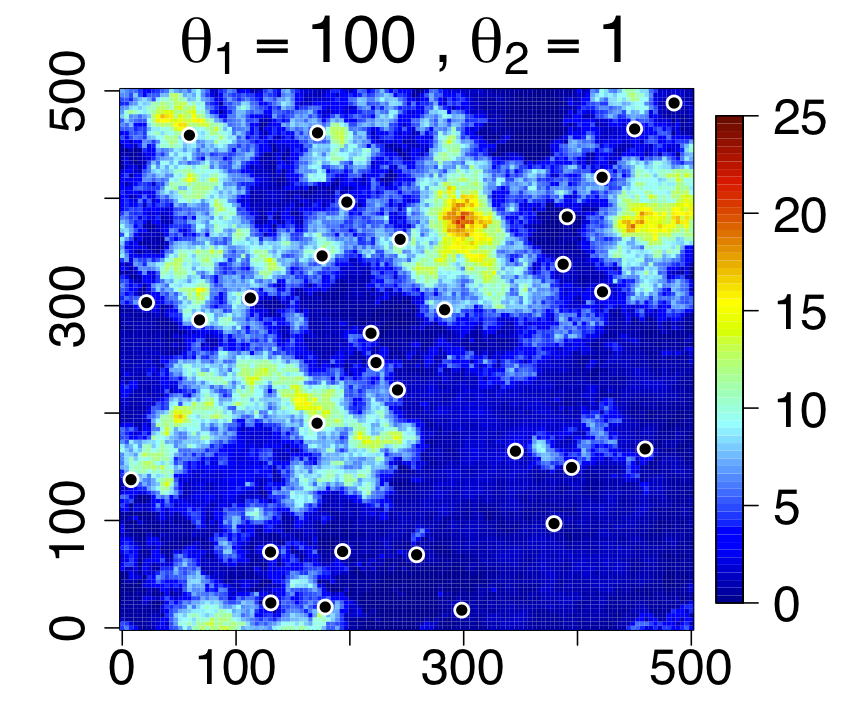}\includegraphics[width=2.25in]{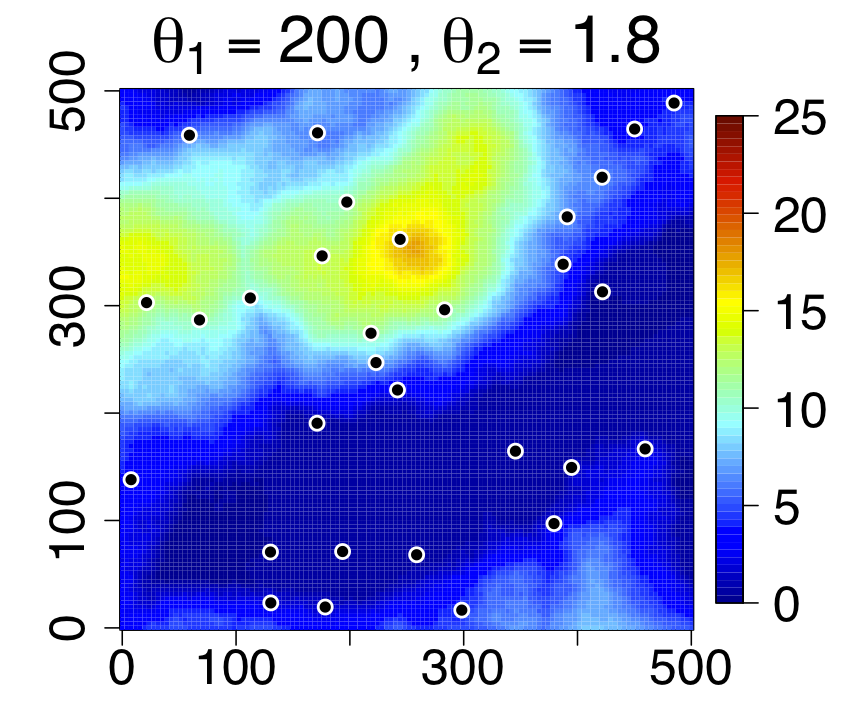}
\par\end{centering}

\begin{centering}
\includegraphics[width=2.25in]{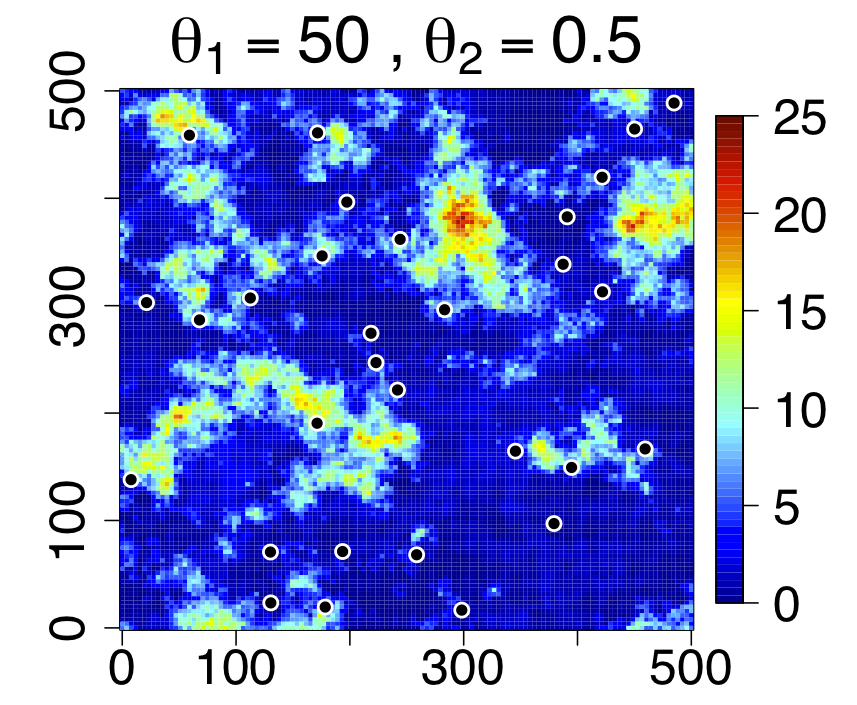}\includegraphics[width=2.25in]{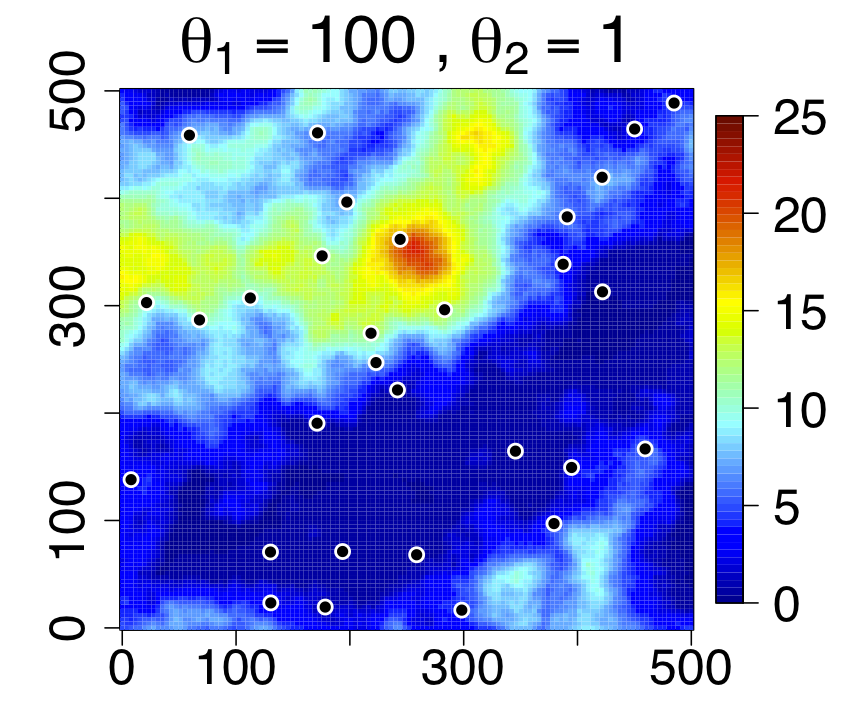}
\par\end{centering}

\begin{centering}
\includegraphics[width=2.25in]{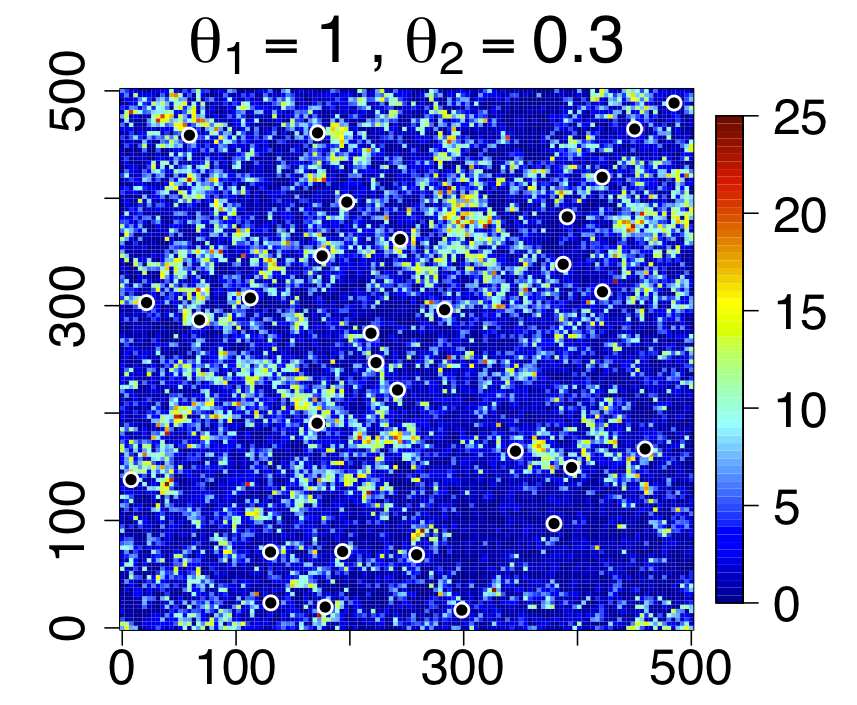}\includegraphics[width=2.25in]{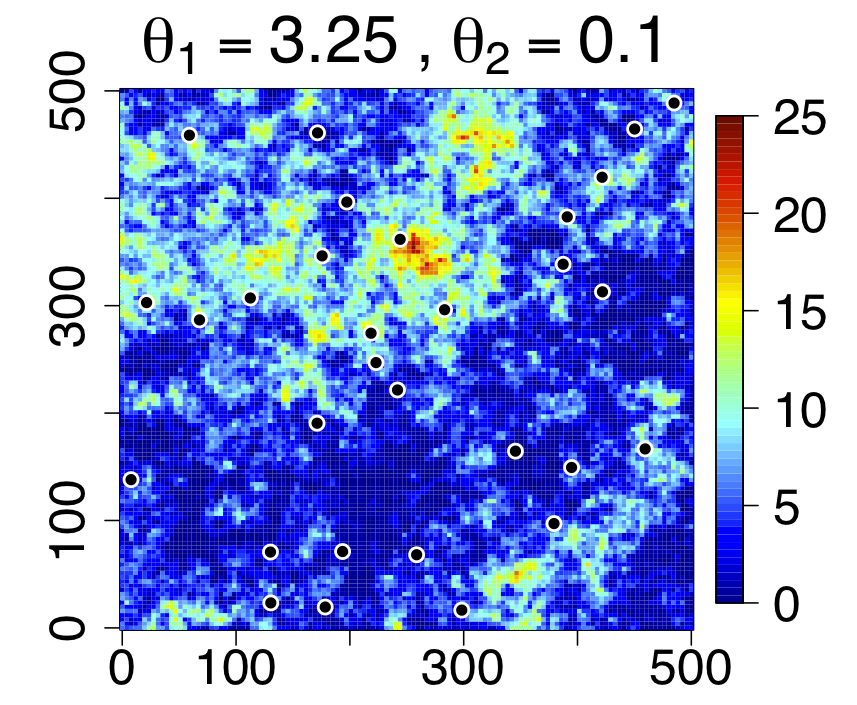}
\par\end{centering}

\end{figure}

\section{Consistency and asymptotic normality of CRPS M-estimators\label{sec:CRPS-M-estimation}}

In this section, we establish general conditions for the consistency
and asymptotic normality of CRPS-based M-estimators. This is motivated
by questions of inference in max-stable models, but may be of independent
interest. Section \ref{sec:CRPS-for-max-stable} implements and specializes
these results to the max-stable setting. 

We start with two theorems that are distillations
of well know results from the general theory of M-estimators, for example
see \cite{vanderVaart1998}. Their proofs are given in Appendix \ref{sub:Asymptotics-proof-appendix}. 
\begin{thm}
\label{thm:CRPSConsistency} Let $X, X^{\left(1\right)}, X^{\left(2\right)}, \ldots$
be iid random vectors with cumulative distribution function $F_{\theta_{0}}.$
Let $\hat{\theta}_{n}$ be as in Definition \ref{def:(CRPS-M-estimator)}
with $\theta_{0}$ an interior point of $\Theta$. Suppose that the
following conditions hold:

\vspace{1pt}

(i) (identifiability) For all $\theta_{1},\theta_{2}\in\Theta,$ 
\begin{equation}
\theta_{1}\not=\theta_{2}\Rightarrow F_{\theta_{1}}\not=F_{\theta_{2}}\quad \text{a.e. } \mu\label{eq:CRPSConsistency1}
\end{equation}

(ii) (integrability) For $B\left(\theta_{0}\right)\subset\Theta,$
an open neighborhood of $\theta_{0}$ 
\begin{equation}
\int_{\mathbb{R}^{d}}\sup_{\theta\in B\left(\theta_{0}\right)}\left(1-F_{\theta}\left(x\right)\right)\mu\left(dx\right)<\infty.\label{eq:CRPSConsistency2}
\end{equation}

(iii) (continuity) The function $\theta\mapsto\int_{\mathbb{R}^{d}}(F_{\theta}(x)-F_{\theta_{0}}(x))^{2}\mu(dx)$
is continuous in the compact parameter space $\Theta\subset\mathbb{R}^{p}$.

Then $\hat{\theta}_{n}\xrightarrow{p}\theta_{0},\ \mbox{as}\ n\to\infty.$ 
\end{thm}

\begin{thm}
\label{thm:(AsyNormCRPS)}Assume the conditions and notation of Theorem
\ref{thm:CRPSConsistency} hold so that in particular, $\hat{\theta}_{n}\xrightarrow{p}\theta_{0}$.
Suppose, moreover, that;

\vspace{1pt}

\begin{itemize}

\item[(i)] The measurable function $\theta\mapsto\mathcal{E}_{\theta}\left(x\right)$
is differentiable at $\theta_{0}$ (for almost every $x$) with gradient
\[
\dot{\mathcal{E}}_{\theta_{0}}\left(x\right):=\left.\frac{\partial}{\partial\theta}\mathcal{E}_{\theta}\left(x\right)\right|_{\theta=\theta_{0}}.
\]

\item[(ii)] There exists a measurable function $L\left(x\right)$
with $\mathbb{E}\left(L\left(X\right)\right)^{2}<\infty$, such that
for every $\theta_{1}$ and $\theta_{2}$ in $B\left(\theta_{0}\right)$
\begin{equation}
\left|\mathcal{E}_{\theta_{1}}\left(x\right)-\mathcal{E}_{\theta_{2}}\left(x\right)\right|\le L\left(x\right)\left\Vert \theta_{1}-\theta_{2}\right\Vert .\label{eq:LipschitzBound}
\end{equation}

\item[(iii)] The map $\theta\mapsto\mathbb{E}\mathcal{E}_{\theta}\left(X\right)$
admits a second-order Taylor expansion at the point of minimum $\theta_{0}$
with non-singular second derivative matrix 
\begin{equation}
H_{\theta_{0}}:=\left.\frac{\partial^{2}}{\partial\theta\partial\theta^{\top}}\mathbb{E}\mathcal{E}_{\theta}\left(X\right)\right|_{\theta=\theta_{0}} .\label{eq:crpsasymbread}
\end{equation}

\end{itemize}

Then

\begin{equation}
\sqrt{n}\left(\hat{\theta}_{n}-\theta_{0}\right)\xrightarrow{d}\mathcal{N}\left(0,H_{\theta_{0}}^{-1}J_{\theta_{0}}H_{\theta_{0}}^{-1}\right),\ \ \mbox{as\ n\ensuremath{\to\infty},}\label{eq:AN}
\end{equation}
 where

\begin{equation}
J_{\theta_{0}}:=\mathbb{E}\left\{ \dot{\mathcal{E}}_{\theta_{0}}\left(X\right)\dot{\mathcal{E}}_{\theta_{0}}\left(X\right)^{\top}\right\} .\label{eq:crpsasymmeat}
\end{equation}

\end{thm}
\medskip{}

The following result provides explicit conditions on the family of
CDFs $\left\{ F_{\theta},\theta\in\Theta\right\} $ that imply conditions
\emph{(i)-(iii)} of Theorem \ref{thm:(AsyNormCRPS)}. It also gives
concrete expressions for the ``bread'' and ``meat'' matrices $H_{\theta_{0}}$
and $J_{\theta_{0}}$ in terms of $F_{\theta}$, which can be used
to compute the asymptotic covariances in \eqref{eq:AN}. The proof is given in Appendix \ref{sub:Asymptotics-proof-appendix}.
\begin{prop}
\label{prop:CRPS-AsymNorm} Assume the conditions and notation in
Theorem \ref{thm:CRPSConsistency}. Suppose moreover that:

\vspace{1pt}

\begin{itemize}

\item[(i)] $\theta\mapsto F_{\theta}\left(y\right)$ is twice continuously
differentiable for all $\theta$ in $B\left(\theta_{0}\right)$ with
gradient $\dot{F}_{\theta}\left(y\right):=\partial F_{\theta}\left(y\right)/\partial\theta$
and second derivative matrix $\ddot{F}_{\theta}\left(y\right):=\partial^{2}F_{\theta}\left(y\right)/\partial\theta\partial\theta^{\top}.$

\item[(ii)] For all $a\in\mathbb{R}^{p}$ with $\left\Vert a\right\Vert >0$
\begin{equation}
\int_{\mathbb{R}^{d}}\left(a^{\top}\dot{F}_{\theta_{0}}\left(y\right)\right)^{2}\mu\left(dy\right)>0.\label{eq:CRPSasymptoticsProp(ii)}
\end{equation}

\item[(iii)] \textup{$\int_{\mathbb{R}^{d}}\sup_{\theta\in B\left(\theta_{0}\right)}\left(\|\dot{F}_{\theta}\left(y\right)\|+\|\dot{F}_{\theta}\left(y\right)\|^{2}+\|\ddot{F}_{\theta}\left(y\right)\|\right)\mu\left(dy\right)<\infty$.}

\end{itemize}

Then (i)-(iii) of Theorem \ref{thm:(AsyNormCRPS)} are satisfied and
therefore \eqref{eq:AN} holds, where 
\begin{equation}
H_{\theta_{0}}:=2\int_{\mathbb{R}^{d}}\dot{F}_{\theta_{0}}\left(y\right)\dot{F}_{\theta_{0}}\left(y\right)^{\top}\mu\left(dy\right)\label{eq:crpsasymbread-1}
\end{equation}
 and
\begin{equation}
J_{\theta_{0}}:=4\int_{\mathbb{R}^{d}}\int_{\mathbb{R}^{d}}\beta_{\theta_{0}}\left(y_{1},y_{2}\right)\dot{F}_{\theta_{0}}\left(y_{1}\right)\dot{F}_{\theta_{0}}\left(y_{1}\right)^{\top}\mu\left(dy_{1}\right)\mu\left(dy_{2}\right)\label{eq:crpsasymmeat-1}
\end{equation}
 where $\beta_{\theta_{0}}\left(y_{1},y_{2}\right)=F_{\theta_{0}}\left(y_{1}\wedge y_{2}\right)-F_{\theta_{0}}\left(y_{1}\right)F_{\theta_{0}}\left(y_{2}\right).$\end{prop}
\begin{rem}
Practical inference utilizing the CRPS M-estimator is limited to cases
where optimization of $\theta\mapsto\mathcal{E}_{\theta}$ is feasible.
Likewise, confidence intervals are only obtained when the matrices
$H_{\theta_{0}}^{-1},J_{\theta_{0}}$ can be computed. Given the multivariate
integration involved, this may require specialized methods for various
models. In the max-stable setting this is achieved through judicious
specification of the measure $\mu$, discussed in the following section.
\end{rem}

\begin{rem}
Condition \eqref{eq:CRPSasymptoticsProp(ii)} ensures that the ``bread''
matrix $H_{\theta_{0}}$ in \eqref{eq:crpsasymbread-1} is non-singular.
It is rather mild and fails only if the gradient $\dot{F}_{\theta_{0}}\left(y\right)$
lies in a lower dimensional hyper-plane for $\mu$-alomost all $y$.
In practice, unless the model is over-parameterized this condition
typically holds.
\end{rem}

\begin{rem}
The expressions \eqref{eq:crpsasymbread-1} and \eqref{eq:crpsasymmeat-1}
can be used in practice to compute the asymptotic covariance matrix
in \eqref{eq:AN}. In Sections \ref{sec:CRPS-for-max-stable} and
\ref{sec:Simulation} we have implemented numerical and Monte Carlo
based methods for calculating $H_{\theta_{0}}$ and $J_{\theta_{0}}$
under the models introduced in Section \ref{sec:Extreme-values-and-nax-stability}. 
\end{rem}

\section{\label{sec:CRPS-for-max-stable}CRPS M-estimation for max-stable
models}

Our goal is to implement the general CRPS method of the previous section
in the case of multivariate max-stable models described in Section
\ref{sec:Extreme-values-and-nax-stability}. Calculation of the CRPS
for such models is aided by a closed form expression of the univariate
CRPS for $1$-Fr\'{e}chet random variates which is given in the
following Lemma.
\begin{lem}
\label{lem:FrechetIntegral} Suppose the measure $\mu$ in Definition
\ref{def:(CRPS-M-estimator)} of the CRPS is specified as $\mu\left(dr\right)=r^{-1/2}dr$
for $r\in\mathbb{R}_{+}.$ Then the univariate CRPS with respect to
the $1$-Fr\'{e}chet distribution function $e^{-v/r}$ has the following
closed form 
\begin{multline}
\mathfrak{F}\left(m,v\right):=\int_{0}^{\infty}\left(e^{-v/r}-\mathbf{1}_{\left\{ m\le r\right\} }\right)^{2}r^{-1/2}dr\\
=4\left[\sqrt{m}\left(e^{-v/m}-\frac{1}{2}\right)+\sqrt{v}\left(\gamma_{\frac{1}{2}}\left(v/m\right)-\sqrt{\frac{\pi}{2}}\right)\right],\label{eq:FrechetIntegral}
\end{multline}
where $\gamma_{\alpha}\left(z\right)=\int_{0}^{z}t^{\alpha-1}e^{-t}dt$
is the incomplete gamma function.
\end{lem}
See Appendix \ref{sub:Proofs-for-Section-4} for a proof. We introduce the notation $\Fr$ to distinguish the univariate
Fr\'{e}chet CRPS from the multivariate case. The functional $\Fr$ is the basis for many of the calculations that follow.

Now recall that the CDF of a $1$-Fr\'{e}chet max-stable random vector
$X=\left(X_{i}\right)_{i=1,\ldots,d}$ is characterized by the tail
function $V\left(x\right)$ as follows
\begin{equation}
F_{X}\left(x\right)=\mathbb{P}\left(X_{i}\le x_{i},i=1,\ldots,d\right)=e^{-V\left(x\right)} ,\label{eq:max-stable-CDF}
\end{equation}
 where $V$ exhibits the homogeneity property $V\left(rx\right)=V\left(x\right)/r$
for all $r>0,x\in(0,\infty]^{d}.$ This means that for any $u=\left(u_{i}\right)_{i=1,\ldots,d}\in\mathbb{R}_{+}^{d}$,
the max-linear combination 
\begin{equation}
M_{u}:=\max_{i=1,\ldots,d}\frac{X_{i}}{u_{i}}\label{eq:maxLinearCombn}
\end{equation}
is a $1$-Fr\'{e}chet variable with scale $V\left(u\right)$. Indeed,
\[
\mathbb{P}\left(M_{u}\le r\right)=\mathbb{P}\left(X_{i}\le ru_{i},i=1,\ldots,d\right)=e^{-V\left(ru\right)}=e^{-V\left(u\right)/r}.
\]
This max-linearity invariance property motivates a particular choice
of the measure $\mu$ that appears in Definition \ref{def:(CRPS-M-estimator)}
for the multivariate CRPS. Let 
\begin{equation}
\mu\left(dy\right)\equiv\mu\left(dr,du\right):=r^{-1/2}dr\sum_{w\in\mathcal{U}}\delta_{w}\left(du\right) ,\label{eq:mumeasure}
\end{equation}
where
$
u=y/\left|y\right|,r=\left|y\right|=\sum_{i=1}^{d}y_{i}\text{ and }\mathcal{U}\subset\mathbb{R}_{+}^{d}.
$
With this choice of $\mu$ we have the following closed form expression
for the multivariate CRPS in terms of the max-linear combinations
$\left\{ M_{u}\right\} _{u\in\mathcal{U}}.$ 
\begin{prop}
With $\mu$ as in \eqref{eq:mumeasure}, for the CRPS in \eqref{eq:CRPS},
we have

\begin{align}
\mathcal{E}_{\theta}\left(X\right) & =\int_{[0,\infty)^{d}}\left[e^{-V_{\theta}\left(y\right)}-\mathbf{1}_{\left\{ X\le y\right\} }\right]^{2}\mu\left(dy\right)\nonumber \\
 & =\sum_{u\in\mathcal{U}}\mathfrak{F}\left(M_{u},V_{\theta}\left(u\right)\right)\label{eq:maxstablecrps3}
\end{align}
 with $\mathfrak{F}$ as in Lemma \ref{lem:FrechetIntegral}.\end{prop}
\begin{proof}
Using the substitution $u=y/\left|y\right|$ and $r=\left|y\right|$,
specifying the measure $\mu$ as in \eqref{eq:mumeasure} results
in
\begin{align*}
\mathcal{E}_{\theta}\left(X\right) & =\int_{[0,\infty)^{d}}\left[e^{-V_{\theta}\left(y\right)}-\mathbf{1}_{\left\{ X\le y\right\} }\right]^{2}\mu\left(dy\right)\\
 & =\sum_{u\in\mathcal{U}}\int_{0}^{\infty}\left[e^{-V_{\theta}\left(ru\right)}-\mathbf{1}_{\left\{ X\le ru\right\} }\right]r^{-1/2}dr.
\end{align*}
Observe that $\{X\le ru\}=\left\{ X_{i}\le ru_{i},i=1,\ldots,d\right\} $
is equivalent to $\left\{ M_{u}\le r\right\} $, where $M_{u}$ is
as in \eqref{eq:maxLinearCombn}. Therefore, using the homogeneity
property $V_{\theta}\left(ru\right)=V_{\theta}\left(u\right)/r$,
we obtain 
\[
\int_{0}^{\infty}\left[e^{-V_{\theta}\left(ur\right)}-\mathbf{1}_{\left\{ X\le ru\right\} }\right]r^{-1/2}dr=\int_{0}^{\infty}\left[e^{-V_{\theta}\left(u\right)/r}-\mathbf{1}_{\left\{ M_{u}\le r\right\} }\right]r^{-1/2}dr.
\]
Lemma \ref{lem:FrechetIntegral} applied to the last integral yields
\eqref{eq:maxstablecrps3}. 
\end{proof}
In practice, given a set of independent observations $X^{\left(1\right)},X^{\left(2\right)},\ldots,X^{\left(n\right)}$
from the model $F_{\theta_{0}}\left(x\right)=\exp\left(-V_{\theta_{0}}\left(x\right)\right)$
we obtain the CRPS-based estimator of $\theta_{0}$ as follows

\begin{center}
\textbf{CRPS estimation procedure}
\par\end{center}
\begin{enumerate}
\item Construct the set $\mathcal{U}\subset\mathbb{R}_{+}^{d}$. The distribution
of $\mathcal{U}$ can be determined heuristically. In general, finite
uniform random samples from the simplex $\Delta^{d-1}:=\left\{ u\in\left(0,1\right)^{d},\left|u\right|=1\right\} $
work well.
\item Construct the max-linear combinations $M_{u}^{\left(i\right)}=\max_{j=1,\ldots,d}X_{j}^{\left(i\right)}/u_{j}$,
for all $i=1,\ldots,n$ and $u\in\mathcal{U}$.
\item Using numerical optimization, compute:
\[
\hat{\theta}_{n}=\underset{\theta\in\Theta}{\arg\min}\sum_{i=1}^{n}\sum_{u\in\mathcal{U}}\mathfrak{F}\left(M_{u}^{\left(i\right)},V_{\theta}\left(u\right)\right).
\]

\end{enumerate}
In Section \ref{sec:Simulation}, we illustrate this methodology over
several concrete examples. The explicit construction of the set $\mathcal{U}$
is given in each example and the computation of the tail dependence
function $V_{\theta}$ when it is not available in closed form is
discussed. 

The following result provides readily computable expressions for the
``bread'' and ``meat'' matrices appearing in the asymptotic covariance
of the CRPS estimators. 
\begin{cor}
\label{cor:HandJ} Using the same specification of the measure $\mu$
as in \eqref{eq:mumeasure} 
\begin{equation}
H_{\theta_{0}}=\frac{\sqrt{\pi}}{2}\sum_{u\in\mathcal{U}}\left(2V_{\theta_{0}}\left(u\right)\right)^{-3/2}\dot{V}_{\theta_{0}}\left(u\right)\left(\dot{V}_{\theta_{0}}\left(u\right)\right)^{\top} \label{eq:Cor1bread}
\end{equation}
 and \textup{
\begin{equation}
J_{\theta_{0}}=\sum_{u,w\in\mathcal{U}}c_{\theta_{0}}\left(u,w\right)\frac{\dot{V}_{\theta_{0}}\left(u\right)\left(\dot{V}_{\theta_{0}}\left(w\right)\right)^{\top}}{\sqrt{V_{\theta_{0}}\left(u\right)V_{\theta_{0}}\left(w\right)}}\label{eq:Cor1meat}
\end{equation}
}where\textup{
\[
c_{\theta_{0}}\left(u,w\right)=\mathrm{Cov}\left\{ \gamma_{\frac{1}{2}}\left(V_{\theta_{0}}\left(u\right)/M_{u}\right),\gamma_{\frac{1}{2}}\left(V_{\theta_{0}}\left(w\right)/M_{w}\right)\right\} .
\]
}\end{cor}
\begin{rem}
$M_{u}$ and $M_{w}$ are dependent since in view of \eqref{eq:maxLinearCombn}
they are defined as max-linear combinations of the vector $X$. The
coefficient $c_{\theta_{0}}\left(u,w\right)$ can be conveniently computed using
Monte Carlo methods by simulating a large number of independent copies
of $X$ under the $F_{\theta_{0}}$ model. In practice the resulting
asymptotic covariance matrix estimates yield confidence intervals
with close to nominal coverage (see Tables \ref{tab:Logistic-model-simulation}
and \ref{tab:Simulation-results-for-Schlather}).
\end{rem}

\section{\label{sec:Simulation} Simulation}

In this section we conduct simulation studies for  CRPS M-estimation under 3 different max-stable models.  
The first example provides a comparison of  CRPS M-estimation to the MLE.  The second example shows that 
CRPS M-estimators can identify dependence structures that are unidentifiable through bivariate distributions only.  
This shows the potential advantages of the new methodology over methods based on partial likelihood.  
The third example illustrates inference for a random field model applicable in spatial extremes

\subsection{Example: multivariate logistic model}

The multivariate logistic is a special case that allows comparison
between our CRPS based estimator and the MLE. This is because the
full joint likelihood is available in this simple model. Hence, we
can estimate the relative efficiency of the CRPS estimator in this
idealized case. To this end, let $\theta=\left(\sigma,\alpha\right)\in\Theta:=(0,\infty)\times(0,1)$
and recall 

\[
V_{\theta}\left(x\right)=\sigma\left(\sum_{t\in D}x_{t}^{-1/\alpha}\right)^{\alpha}
\]
 is the tail dependence function of a multivariate logistic max-stable
model. We estimate the parameters for the model when $\left|D\right|=5$
and $\theta_{0}=(5,0.7)$, using samples sizes $n=100$ and $n=1000$
with $500$ replications each. Realizations were generated using the
{\tt R} package {\tt evd} \citep{evd-Rpackage}. For each realization $X^{\left(i\right)},i=1,\ldots,n$
we construct the max-linear combinations $M_{u}^{\left(i\right)}$
using a (fixed) uniform sample $\mathcal{U}\subset\Delta^{d-1}$ where
$\left|\mathcal{U}\right|=1000.$ Numerical optimization of the CRPS
criterion in \eqref{eq:maxstablecrps3} was carried out using {\tt R}'s
{\tt optim} routine with an arbitrary starting point in the interior
of $\Theta$. Results for both the CRPS estimators and the MLE are
shown in Table \ref{tab:Logistic-model-simulation}. \renewcommand{\arraystretch}{1.5}
\begin{table}
\caption{Logistic model simulation results using 500 replications. Reported
are the empirical mean and standard deviation of the CRPS and (MLE)
estimates. Coverages are based on plug-in estimates of 95\% asymptotic
confidence intervals. In the case of the CRPS estimates, confidence
intervals are generated using the expressions from Corollary \ref{cor:HandJ}.
\label{tab:Logistic-model-simulation} }

\centering{}%
\begin{tabular}{ccc|cc}
 & \multicolumn{4}{c}{CRPS (MLE)}\tabularnewline
 & \multicolumn{2}{c|}{$n=100$} & \multicolumn{2}{c}{$n=1000$}\tabularnewline
\cline{2-5} 
 & $\sigma(5)$ & $\alpha(0.7)$ & $\sigma(5)$ & $\alpha(0.7)$\tabularnewline
\cline{2-5} 
mean & 5.000 (5.024) & 0.700 (0.699) & 4.999 (5.009) & 0.701 (7.000)\tabularnewline
\cline{2-5} 
sd & 0.519 (0.319) & 0.048 (0.028) & 0.158 (0.100) & 0.015 (.008)\tabularnewline
\cline{2-5} 
.95 coverage & 0.934 (0.962) & 0.940 (0.910) & 0.954 (0.948)  & 0.952 (0.956)\tabularnewline
\end{tabular}
\end{table}
\renewcommand{\arraystretch}{1} Observe that we have essentially
unbiased estimators. The asymptotic confidence intervals based on
\eqref{eq:AN} were computed using the expressions in Corollary \ref{cor:HandJ}
and have close to nominal coverages even for moderate sample size
$n=100.$ As expected, the CRPS is less efficient than the MLE however,
the results in Table \ref{tab:Logistic-model-simulation} provide
evidence that suggest the CRPS is a good alternative when the MLE
is not available as is the case with the remaining examples.

\subsection{Example: Max-linear model}

Let $d=3$ and $k=4$ and define two $\left(d\times k\right)$ matrices
\[
B=\left(\begin{array}{cccc}
1 & 1 & 0 & 0\\
1 & 0 & 1 & 0\\
0 & 1 & 1 & 0
\end{array}\right)  \text{ and }  C=\left(\begin{array}{cccc}
1 & 1 & 0 & 0\\
1 & 0 & 1 & 0\\
1 & 0 & 0 & 1
\end{array}\right).
\]
Let $Z_{1},\ldots,Z_{4}$ be iid $1$-Fr\'{e}chet random variables
and define 
\begin{equation}
X_{i}=\max_{j=1,\ldots,k}a_{ij}Z_{j},\label{eq:max-linear-model}
\end{equation}
where 
\[
\left(a_{ij}\left(\theta\right)\right)=A\left(\theta\right)=\theta B+\left(1-\theta\right)C,\ \theta\in\left\{ 0,1\right\} .
\]
The tail dependence function for this model is 
\[
V_{\theta}\left(x\right)=\sum_{j=1}^{k}\max_{i=1,\ldots,d}a_{ij}\left(\theta\right)/x_{i}.
\]
We simulated $500$ replications from the max-linear model \eqref{eq:max-linear-model}
with $\theta_{0}=1$. For each realization $X^{\left(i\right)},i=1,\ldots,n$
we construct the max-linear combinations $M_{u}^{\left(i\right)}$
using a random uniform sample $\mathcal{U}\subset\Delta^{d-1}$ where
$\left|\mathcal{U}\right|=1000.$ We estimate $\theta_{0}$ via the
CRPS estimator 
\begin{equation}
\hat{\theta}_{n}=\underset{\theta\in\left\{ 0,1\right\} }{\arg\min}\sum_{i=1}^{n}\sum_{u\in\mathcal{U}}\mathfrak{F}\left(M_{u}^{\left(i\right)},V_{\theta}\left(u\right)\right)\label{eq:max-linear-CRPS}
\end{equation}
There is no need for numerical optimization in this case since we
can calculate the CRPS under $\theta=1$ and $\theta=0.$ Results
in Table \ref{tab:Max-linear-model-classification} show that the
error rate for $\hat{\theta}_{n}=\theta_{0}$ decreases as the sample
size $n$ increases.

\renewcommand{\arraystretch}{1.5}

\begin{table}
\caption{Error rate based on 500 replications of the CRPS estimator \eqref{eq:max-linear-CRPS}
for max-linear model \eqref{eq:max-linear-model}. \label{tab:Max-linear-model-classification}}

\begin{centering}
\begin{tabular}{cc|c|c}
 & $n=100$ & $n=500$ & $n=1000$\tabularnewline
\cline{2-4} 
Error rate  & 0.332 & 0.154 & 0.066\tabularnewline
\end{tabular}
\par\end{centering}

\end{table}

\renewcommand{\arraystretch}{1}
\begin{rem}
When considering the marginal structure 
\[
\begin{array}{ccc}
\theta=1: &  & \theta=0:\\
X_{1}=Z_{1}\vee Z_{2} &  & X_{1}=Z_{1}\vee Z_{2}\\
X_{2}=Z_{1}\vee Z_{3} &  & X_{2}=Z_{1}\vee Z_{3}\\
X_{3}=Z_{2}\vee Z_{3} &  & X_{3}=Z_{1}\vee Z_{4}
\end{array}
\]
the bivariate and univariate marginals are equal under $\theta=0$ or $\theta=1.$ Hence, the parameter $\theta$ is unidentifiable from
statistics based on bivariate distributions.  Because the CRPS relies on the full
joint distribution of $X$, it is able to discriminate between the
two models. One can similarly construct different max-linear models that have equal $k$-dimensional distributions for $k\le d$.
\end{rem}

\subsection{Example: Schlather model}

We now provide an example that is applicable in the spatial setting.
Let $\left\{ w_{t}\right\} _{t\in T}$ be a Gaussian process on $T\subset\mathbb{R}^{2}$
with standard normal margins and let $\rho_{\theta}\left(t,s\right)$
be its associated correlation function parameterized by $\theta$. Define

\[
V_{\theta}\left(x\right)=\mathbb{E}_{\theta}\max_{t\in D}\left\{ \left[\sqrt{2\pi}w_{t}\right]_{+}/x_{t}\right\} 
\]
then $V_{\theta}\left(x\right)$ is the tail dependence function of
a Schlather max-stable model with standard $1$-Fr\'{e}chet marginals, where the process is observed at a set
of locations $D$. In this case $V_{\theta}\left(x\right)$ is not
available in closed form, instead we use a Monte Carlo approximation
from a large sample $w_{t}^{\left(i\right)}$, $i=1,\ldots,K$ under
$\theta$. For this simulation we assume a \emph{stable }correlation
function, i.e.
\[
\rho_{\theta}\left(t,s\right)=\exp\left[-\left(\left\Vert t-s\right\Vert /\sigma\right)^{\alpha}\right],\ \theta=\left(\sigma,\alpha\right)\in\Theta=\left(0,\infty\right)\times(0,2].
\]
The top row of Figure \ref{fig:max-stable-realizations} shows realizations
from this Schlather model under two different parameter settings.
For our study we set $\theta_{0}=\left(100,1\right)$ and simulated
$100$ replications at $d=30$ uniformly sampled locations over a
$500\times500$ grid. This corresponds to the top left panel in
Figure \ref{fig:max-stable-realizations}. 

Realizations were generated using the {\tt R} package {\tt SpatialExtremes} \citep{Ribatet:spatialextremes-Rpackage}.
For each realization $X^{\left(i\right)},i=1,\ldots,n$ we construct
the max-linear combinations $M_{u}^{\left(i\right)}$ using a random
uniform sample $\mathcal{U}\subset\Delta^{d-1}$, where $\left|\mathcal{U}\right|=1000.$
For sample sizes $n=100$ and $n=1000$, we numerically optimize the CRPS
criterion \eqref{eq:maxstablecrps3} using {\tt R}'s {\tt optim}
routine with multiple starting points in the interior of $\Theta$.
Simulation results in Table \ref{tab:Simulation-results-for-Schlather}
show the CRPS estimates are essentially unbiased and display close to nominal
coverage. For comparison we also provide pairwise MCLE estimates fitted using the 
{\tt SpatialExtremes} package.  For information on pairwise MCLE see \cite{PadoanRibatetSisson2010}.
\renewcommand{\arraystretch}{1.5}

\begin{table}
\caption{CRPS and MCLE estimates for Schlather model. Reported are mean and standard
deviation of 100 replications using sample size $n=100$ and $n=100$.
CRPS based confidence intervals for $\theta_{0}=\left(100,1\right)$ were calculated
using plug-in estimates for the expressions in Corollary \ref{cor:HandJ} and resulting
$95\%$ coverages are reported.  Coverages for MCLE estimates are based on sandwich estimators of
\cite{PadoanRibatetSisson2010}. \label{tab:Simulation-results-for-Schlather}}

\centering{}%
\begin{tabular}{ccc|cc}
 & \multicolumn{4}{c}{CRPS (MCLE)}\tabularnewline
 & \multicolumn{2}{c|}{$n=100$} & \multicolumn{2}{c}{$n=500$}\tabularnewline
\cline{2-5} 
 & $\theta_{1}\left(100\right)$ & $\theta_{2}\left(1\right)$ & $\theta_{1}\left(100\right)$ & $\theta_{2}\left(1\right)$\tabularnewline
\cline{2-5} reference
mean & 110.56 (99.80) & 1.25 (1.01) & 99.71 (100.50)  & 1.10 (1.00)\tabularnewline
\cline{2-5} 
sd & 113.73 (14.92) &  0.63 (0.18) & 45.67 (7.01) & 0.42 (0.08)\tabularnewline
\cline{2-5} 
.95 coverage & 0.98 (0.95) & 0.90 (0.93) & 0.96 (0.95) & 0.92 (0.94)\tabularnewline
\end{tabular}
\end{table}
\renewcommand{\arraystretch}{1}

Note that for this model, the two estimators are comparable in terms of bias but the MCLE is more efficient than CRPS.  This may indicate that MCLE is especially well suited for
estimation with spectrally Gaussian max-stable models.

\section{\label{sec:discussion}Discussion}

We have developed a general inferential framework for max-stable models
based on the continuous ranked probability score (CRPS). It is shown
that under mild regularity, CRPS M-estimators are consistent and asymptotically normal. Simulation
studies across common spectrally continuous and discrete max-stable
models yield essentially unbiased estimators with close to nominal coverage.
Our estimators were about half as efficient versus the MLE in the
case of the simple multivariate logistic model, where a tractable likelihood
exists. Overall the method displays flexibility and broad applicability
in the max-stable setting. 

In the case of the Schlather max-stable model, CRPS estimates were less efficient than MCLE.  It is possible that efficiency for CRPS estimates
can be improved through better tuning of the measure $\mu$ in the
CRPS. For instance consider 
\[
\mu\left(dr,du\right)=r^{-\eta}\sum_{w\in\mathcal{U}}\delta_{\{w\}}\left(du\right).
\]
It can be shown that CRPS M-estimation remains consistent for all
$\eta > 0$ with little complication over
the case $\eta=1/2$ (equivalent to the specification \eqref{eq:mumeasure}), 
which was chosen for analytical simplicity.
This begs the question of specifying  $\eta$ to maximize the
expected Hessian of the CRPS, which should result in more efficient estimators.
This is beyond the scope of the present paper and it will be studied in a future work.

\appendix

\section{Proofs}

\subsection{\label{sub:Derivation-of-max-stableCDF} De Haan's spectral representation}

For completeness, we provide next the formal proof of the Poisson point representation due to de Haan.
For more details see \cite{deHaan1984,stoev:taqqu:2005,kabluchko:2009}.
 
\begin{proof}[Proof of Proposition \ref{p:de-Haan-simple}]
By \eqref{eq:maxLinRF}, for all $x_i>0,\ i=1,\cdots,d$,
$$
\P (X_{t_i}\le x_i,\ i=1,\cdots,d ) = \P( \Pi \subset A ) = \P(\Pi \cap A^c = \emptyset),
$$
where 
$$
 A = \{ (u,s) \in \R_+\times S\, :\, g_{t_i}(s)/u \le x_i,\ i=1,\cdots,d\}.
$$
Observe that 
$
A^c = \{ (u,s) \, :\, \max_{i=1,\cdots,d} g_{t_i}(s)/x_i > u \}. 
$
Since $\Pi$ is a Poisson point process on $\R_+\times S$ with intensity $du \nu(ds)$,
$$
\P (\Pi \cap A^c = \emptyset) = \exp{\Big\{} - \int_{S} \int_0^{\max_{i=1,\cdots,d} g_{t_i}(s)/x_i } du \nu(ds){\Big\}},
$$
which equals \eqref{e:fdd-de-Haan} and completes the proof.
The above argument shows that the integrability of the functions $g_t$ implies the $X_t$s in \eqref{eq:maxLinRF} 
are non-trivial random variables.
\end{proof}

%
%

\subsection{\label{sub:Asymptotics-proof-appendix} Proofs for Section \ref{sec:CRPS-M-estimation}}
\begin{proof}[Proof of Theorem \ref{thm:CRPSConsistency}]

Observe that the estimator $\hat{\theta}_{n}$ in Definition \ref{def:(CRPS-M-estimator)}
trivially satisfies $n^{-1}\sum_{i=1}^{n}\mathcal{E}_{\hat{\theta}_{n}}\left(X_{i}\right)\le n^{-1}\sum_{i=1}^{n}\mathcal{E}_{\theta_{0}}\left(X_{i}\right)-o_{P}\left(1\right)$.
Therefore, by \citealp[Thm. 5.7 of ][]{vanderVaart1998}, the desired
consistency follows if 
\begin{equation}
\sup_{\theta\in\Theta}\left|\frac{1}{n}\sum_{i=1}^{n}\mathcal{E}_{\theta}\left(X_{i}\right)-\mathbb{E}\mathcal{E}_{\theta}\left(X\right)\right|\xrightarrow{P}0\label{eq:MestConsistency1}
\end{equation}
 and

\begin{equation}
\sup_{\theta:\left\Vert \theta-\theta_{0}\right\Vert \ge\epsilon,\ \theta\in\Theta}\mathbb{E}\mathcal{E}_{\theta}\left(X\right)>\mathbb{E}\mathcal{E}_{\theta_{0}}\left(X\right),\ \ \mbox{for all\ }\epsilon>0.\label{eq:MestConsistency2}
\end{equation}

We will first show \eqref{eq:MestConsistency2}. By Fubini's Theorem,
we have

\begin{multline}
\mathbb{E}\mathcal{E}_{\theta}\left(x\right)=\int_{\mathbb{R}^{d}}\left(F_{\theta}\left(y\right)-F_{\theta_{0}}\left(y\right)\right)^{2}\mu\left(dy\right)\\
+\int_{\mathbb{R}^{d}}F_{\theta_{0}}\left(y\right)\left(1-F_{\theta_{0}}\left(y\right)\right)\mu\left(dy\right)\\
\ge\int_{\mathbb{R}^{d}}F_{\theta_{0}}\left(y\right)\left(1-F_{\theta_{0}}\left(y\right)\right)\mu\left(dy\right)=\mathbb{E}\mathcal{E}_{\theta_{0}}\left(x\right).\label{eq:EcrpsExpansion}
\end{multline}

This implies \eqref{eq:MestConsistency2} because the continuity condition
\emph{(iii) } and the compactness of $\Theta$ gaurantee the supremum therein is attained for some
$\theta^{\ast}\not=\theta_{0}$. 

We now show \eqref{eq:MestConsistency1}. Let $F_{n}\left(x\right)=n^{-1}\sum_{i=1}^{n}\mathbf{1}\left\{ X^{\left(i\right)}\le x\right\} $
and $\overline{F}=1-F$. Note that

\begin{multline}
\sup_{\theta\in\Theta}\left|\frac{1}{n}\sum_{i=1}^{n}\mathcal{E}_{\theta}\left(X^{\left(i\right)}\right)-\mathbb{E}\mathcal{E}_{\theta}\left(X\right)\right|\\
=\sup_{\theta\in\Theta}\left|\int_{\mathbb{R}^{d}}\left(1-2F_{\theta}\left(x\right)\right)\left(F_{n}\left(x\right)-F_{\theta_{0}}\left(x\right)\right)\mu\left(dx\right)\right|\\
\le2\int_{\mathbb{R}^{d}}\left|F_{n}\left(x\right)-F_{\theta_{0}}\left(x\right)\right|\mu\left(dx\right).\label{eq:Consistency_Proof_3}
\end{multline}
 Fix $\epsilon>0.$ Markov's inequality and another application of
Fubini gives
\begin{equation}
\mathbb{P}\left\{ \int_{\mathbb{R}^{d}}\left|F_{n}\left(x\right)-F_{\theta_{0}}\left(x\right)\right|\mu\left(dx\right)>\epsilon\right\} \\
\le \frac{1}{\epsilon}\int_{\mathbb{R}^{d}}\mathbb{E}\left|\overline{F}_{n}\left(x\right)-\overline{F}_{\theta_{0}}\left(x\right)\right|\mu\left(dx\right).\label{eq:Consistency_Proof_2}
\end{equation}
 Next, using the identity $\left|a-b\right|=a+b-2a\wedge b$ we have
that the RHS of \eqref{eq:Consistency_Proof_2} equals
\begin{multline}
\frac{1}{\epsilon}\int_{\mathbb{R}^{d}}\mathbb{E}\left\{ \overline{F}_{n}\left(x\right)+\overline{F}_{\theta_{0}}\left(x\right)-2\overline{F}_{n}\left(x\right)\wedge\overline{F}_{\theta_{0}}\left(x\right)\right\} \mu\left(dx\right)\\
=\frac{2}{\epsilon}\left\{ \int_{\mathbb{R}^{d}}\overline{F}_{\theta_{0}}\left(x\right)\mu\left(dx\right)-\int_{\mathbb{R}^{d}}\mathbb{E}\left[\overline{F}_{n}\left(x\right)\wedge\overline{F}_{\theta_{0}}\left(x\right)\right]\mu\left(dx\right)\right\} \label{eq:Consistency_Proof_1}
\end{multline}
 Note that $\mathbb{E}\left[\overline{F}_{n}\left(x\right)\wedge\overline{F}_{\theta_{0}}\left(x\right)\right]\le\overline{F}_{\theta_{0}}\left(x\right),$
and by condition \emph{(ii)}, $\int_{\mathbb{R}^{d}}\overline{F}_{\theta_{0}}\left(x\right)\mu\left(dx\right)<\infty.$
Thus, by the Lebesgue dominated convergence theorem 
\[
\lim_{n\to\infty}\int_{\mathbb{R}^{d}}\mathbb{E}\left[\overline{F}_{n}\left(x\right)\wedge\overline{F}_{\theta_{0}}\left(x\right)\right]\mu\left(dx\right)=\int_{\mathbb{R}^{d}}\lim_{n\to\infty}\mathbb{E}\left[\overline{F}_{n}\left(x\right)\wedge\overline{F}_{\theta_{0}}\left(x\right)\right]\mu\left(dx\right).
\]
 The strong law of large numbers implies that $\overline{F}_{n}\left(x\right)\wedge\overline{F}_{\theta_{0}}\left(x\right)$
converges almost surely to $\overline{F}_{\theta_{0}}\left(x\right)\wedge\overline{F}_{\theta_{0}}\left(x\right)\equiv\overline{F}_{\theta_{0}}\left(x\right)$.
Hence, by applying dominated convergence again, we obtain 
\[
\lim_{n\to\infty}\mathbb{E}\left[\overline{F}_{n}\left(x\right)\wedge\overline{F}_{\theta_{0}}\left(x\right)\right]=\overline{F}_{\theta_{0}}\left(x\right),\ \ \mbox{for all }x\in\mathbb{R}^{d}.
\]
 This, by \eqref{eq:Consistency_Proof_1} implies that the right-hand
side of \eqref{eq:Consistency_Proof_2} vanishes as $n\to\infty$,
which in view of \eqref{eq:Consistency_Proof_3} yields the desired
convergence in probability \eqref{eq:MestConsistency1} and the proof
is complete.
\end{proof}

\begin{proof}[Proof of Theorem \ref{thm:(AsyNormCRPS)}]

Since the CRPS estimator $\hat{\theta}_{n}$ minimizes the CRPS distance,
we trivially have $n^{-1}\sum_{i=1}^{n}\mathcal{E}_{\hat{\theta}_{n}}\left(X_{i}\right)\le n^{-1}\sum_{i=1}^{n}\mathcal{E}_{\theta_{0}}\left(X_{i}\right)-o_{P}\left(n^{-1}\right)$.
Thus, by \citealp[Thm. 5.23 of ][]{vanderVaart1998} the asymptotic
normality in \eqref{eq:AN} follows, provided conditions (i)-(iii)
hold.
\end{proof}

\begin{proof}[Proof of Proposition \ref{prop:CRPS-AsymNorm}]
\label{pf:Proof-CRPS-AsymNormality-1} 

By a standard argument using the Lebesgue DCT, condition \emph{(iii)
}of this proposition ensures that integration and differentiation
can be interchanged in all that follows. We proceed by establishing
\emph{(i)-(iii) }of Theorem \ref{thm:(AsyNormCRPS)}. 

\emph{(i) }By the differentiability of $\theta\mapsto F_{\theta}$
for all $\theta\in B\left(\theta_{0}\right)$ the function $\theta\mapsto\mathcal{E}_{\theta}$
is differentiable at $\theta_{0}$ since exchanging integration and
differentiation allows
\begin{align*}
\dot{\mathcal{E}}_{\theta_{0}} & =\frac{\partial}{\partial\theta}\left.\int_{\mathbb{R}^{d}}\left(F_{\theta}\left(y\right)-\mathbf{1}\left\{ x\le y\right\} \right)^{2}\mu\left(dy\right)\right|_{\theta=\theta_{0}}\\
 & =2\int_{\mathbb{R}^{d}}\left(F_{\theta}\left(y\right)-\mathbf{1}\left\{ x\le y\right\} \right)\dot{F}_{\theta}\left(y\right)\mu\left(dy\right).
\end{align*}

\emph{(ii) }Observe that $\left|\mathcal{E}_{\theta_{1}}\left(x\right)-\mathcal{E}_{\theta_{2}}\left(x\right)\right|$
equals 
\begin{multline*}
\left|\int_{\mathbb{R}^{d}}\left\{ \left(F_{\theta_{1}}\left(y\right)-\mathbf{1}\left\{ x\le y\right\} \right)^{2}-\left(F_{\theta_{2}}\left(y\right)-\mathbf{1}\left\{ x\le y\right\} \right)^{2}\right\} \mu\left(dy\right)\right|\\
=\left|\int_{\mathbb{R}^{d}}\left\{ \left[\left(F_{\theta_{1}}\left(y\right)+F_{\theta_{2}}\left(y\right)\right)-2\mathbf{1}\left\{ x\le y\right\} \right]\left(F_{\theta_{1}}\left(y\right)-F_{\theta_{2}}\left(y\right)\right)\right\} \mu\left(dy\right)\right|\\
\le2\int_{\mathbb{R}^{d}}\left|F_{\theta_{1}}\left(y\right)-F_{\theta_{2}}\left(y\right)\right|\mu\left(dy\right)
\end{multline*}
where the last relation follows from the triangle inequality and fact
that $\left|F_{\theta}\left(y\right)-\mathbf{1}\left\{ x\le y\right\} \right|\le\max\left\{ F_{\theta}\left(y\right),1-F_{\theta}\left(y\right)\right\} \le1.$
Then, by the mean value theorem and the Cauchy-Schwartz inequality
\begin{align}
\int_{\mathbb{R}^{d}}\left|F_{\theta_{1}}\left(y\right)-F_{\theta_{2}}\left(y\right)\right|\mu\left(dy\right) & \le\left\Vert \theta_{1}-\theta_{2}\right\Vert \int_{\mathbb{R}^{d}}\sup_{\theta\in B\left(\theta_{0}\right)}\left\Vert \dot{F}_{\theta}\left(y\right)\right\Vert \mu\left(dy\right)\label{eq:L-bound}\\
 & \equiv L\left\Vert \theta_{1}-\theta_{2}\right\Vert \nonumber 
\end{align}
where $L:=\int_{\mathbb{R}^{d}}\sup_{\theta\in B\left(\theta_{0}\right)}\left\Vert \dot{F}_{\theta}\left(y\right)\right\Vert \mu\left(dy\right).$
By assumption \emph{(ii) }of this proposition, $L$ is finite. Hence
\emph{(ii) }of Theorem \ref{thm:(AsyNormCRPS)} holds where $L\left(X\right)\equiv L$
is constant (and therefore trivially $\mathbb{E}\left(L\left(X\right)^{2}\right)<\infty$). 

\emph{(iii) }Existence of a second order Taylor expansion for $\theta\mapsto\mathbb{E}\mathcal{E}_{\theta}\left(X\right)$
follows from the twice continuous differentiability of $\theta\mapsto F_{\theta}$
for all $\theta\in B\left(\theta_{0}\right)$ by

\begin{eqnarray*}
\frac{\partial^{2}}{\partial\theta\partial\theta^{\top}}\mathbb{E}\mathcal{E}_{\theta}\left(X\right) & \overset{\eqref{eq:EcrpsExpansion}}{=} & \frac{\partial^{2}}{\partial\theta\partial\theta^{\top}}\int_{\mathbb{R}^{d}}\left(F_{\theta}\left(y\right)-F_{\theta_{0}}\left(y\right)\right)^{2}\mu\left(dy\right)\\
 & = & \int_{\mathbb{R}^{d}}\frac{\partial^{2}}{\partial\theta\partial\theta^{\top}}\left(F_{\theta}\left(y\right)-F_{\theta_{0}}\left(y\right)\right)^{2}\mu\left(dy\right).
\end{eqnarray*}
The above display implies that
\begin{multline*}
H_{\theta_{0}}=\int_{\mathbb{R}^{d}}\left.\frac{\partial^{2}}{\partial\theta\partial\theta^{\top}}\left(F_{\theta}\left(y\right)-F_{\theta_{0}}\left(y\right)\right)^{2}\right|_{\theta=\theta_{0}}\mu\left(dy\right)\\
=2\int_{\mathbb{R}^{d}}\dot{F}_{\theta_{0}}\left(y\right)\dot{F}_{\theta_{0}}\left(y\right)^{\top}\mu\left(dy\right)=\eqref{eq:crpsasymbread-1}
\end{multline*}
 where non-singularity of $H_{\theta_{0}}$ follows from \emph{(ii)
}because for all $a\in\mathbb{R}^{p}$ with $\left\Vert a\right\Vert >0$
\[
a^{\top}H_{\theta_{0}}a=2\int_{\mathbb{R}^{d}}\left[a^{\top}\dot{F}\left(y\right)\right]^{2}\mu\left(dy\right)>0.
\]

Finally, we derive $J_{\theta_{0}}$ by considering its $ij$th entry.
Let $\partial_{i}$ denote $\partial/\partial\theta_{i}$. 
\begin{eqnarray*}
\left(J_{\theta_{0}}\right)_{ij} & = & \mathbb{E}\left[\left.\partial_{i}\mathcal{E}_{\theta}\left(X\right)\partial_{j}\mathcal{E}_{\theta}\left(X\right)\right|_{\theta=\theta_{0}}\right]\\
 & = & \mathbb{E}\left\{ \int_{\mathbb{R}^{d}}2\left(F_{\theta}\left(y_{1}\right)-\mathbf{1}_{\left\{ X\le y_{1}\right\} }\right)\partial_{i}F_{\theta}\left(y_{1}\right)\mu\left(dy_{1}\right)\right.\\
 &  & \qquad\times\left.\left.\int_{\mathbb{R}^{d}}2\left(F_{\theta}\left(y_{2}\right)-\mathbf{1}_{\left\{ X\le y_{2}\right\} }\right)\partial_{j}F_{\theta}\left(y_{2}\right)\mu\left(dy_{2}\right)\right|_{\theta=\theta_{0}}\right\} \\
 & = & 4\mathbb{E}\left\{ \left.\int_{\mathbb{R}^{d}}\int_{\mathbb{R}^{d}}b_{\theta}\left(X,y_{1},y_{2}\right)\partial_{i}F_{\theta}\left(y_{1}\right)\partial_{j}F_{\theta}\left(y_{2}\right)\mu\left(dy_{1}\right)\mu\left(dy_{2}\right)\right|_{\theta=\theta_{0}}\right\} 
\end{eqnarray*}
where $b_{\theta}\left(X,y_{1},y_{2}\right)=\left(\mathbf{1}_{\left\{ X\le y_{1}\right\} }-F_{\theta}\left(y_{1}\right)\right)\left(\mathbf{1}_{\left\{ X\le y_{2}\right\} }-F_{\theta}\left(y_{2}\right)\right).$
Expanding the integrand and applying Fubini gives 
\[
\left(J_{\theta_{0}}\right)_{ij}=4\int_{\mathbb{R}^{d}}\int_{\mathbb{R}^{d}}\beta_{\theta_{0}}\left(y_{1},y_{2}\right)\partial_{i}F_{\theta_{0}}\left(y_{1}\right)\partial_{j}F_{\theta_{0}}\left(y_{2}\right)\mu\left(dy_{1}\right)\mu\left(dy_{2}\right)
\]
where $\beta_{\theta_{0}}\left(y_{1},y_{2}\right)=\mathbb{E}b_{\theta_0}\left(X,y_{1},y_{2}\right)=F_{\theta_{0}}\left(y_{1}\wedge y_{2}\right)-F_{\theta_{0}}\left(y_{1}\right)F_{\theta_{0}}\left(y_{2}\right)$
which is exactly the $ij$th element of \eqref{eq:crpsasymmeat-1},
as desired. 
\end{proof}

\subsection{Proofs for Section \ref{sec:CRPS-for-max-stable}\label{sub:Proofs-for-Section-4}}
\begin{proof}[Proof of Lemma \ref{lem:FrechetIntegral}]

Observe that $e^{-v/s} - \mathbf{1}_{\{ x\le s\}} = e^{-v/s}$ for $0<s<x$ and hence the integrand  in \eqref{eq:FrechetIntegral} vanishes, as $s\to 0$.  Also, by using a Taylor series expansion of the exponential  function at zero, it is easy to see that $(e^{-v/s} - \mathbf{1}_{\{ x\le s\}}) = (e^{-v/s} - 1) \sim -v/s,$ as $s\to\infty$.  Therefore, the integral in \eqref{eq:FrechetIntegral} is finite.

We have that 

$$ \Fr (x,v) = \int_0^x e^{-2v/s} s^{-1/2} ds + \int_x^\infty (e^{-v/s}-1)^{2} s^{-1/2} ds =: I_1 + I_2.  $$ 
Using that $(2\sqrt{s})'=1/\sqrt{s}$ and integration by parts in both integrals, we obtain 

$$ I_1 = 2\sqrt{x}e^{-2v/x} - 4v \int_0^x s^{-1/2 - 1} e^{-2v/s} ds  $$ 
and 
$$ I_2 = -\sqrt{x}(e^{-v/x} -1)^2 - 4v \int_x^{\infty} s^{-1/2-1} (e^{-2v/s} - e^{v/s}) ds. $$ 
Routine manipulations yield 

\begin{equation}\label{e:CRPS-delta-1} 
I_1 + I_2 = 2\sqrt{x}(2e^{-v/x} -1) + 4v {\Big(} \underbrace{\int_x^\infty s^{-1/2-1} e^{v/s} ds}_{=:J_1} -     \underbrace{\int_0^\infty s^{-1/2-1} e^{-2v/s}  ds}_{=:J_2} {\Big)} 
\end{equation} 

Now, by making the changes of variables $y = v/s$ and $z=2v/s$ in the last two integrals respectively, we obtain \begin{eqnarray*}  J_1 - J_2  &=& v^{-1/2} \int_0^{-v/x}y^{1/2-1}e^{-y}dy - (2v)^{-1/2}\int_0^\infty z^{1/2 -1}e^{-z} dz.
\end{eqnarray*}    
This, in view of \eqref{e:CRPS-delta-1}, yields the expression in terms of the incomplete gamma  function in \eqref{eq:FrechetIntegral}. 
\end{proof}
The proof of Corollary \ref{cor:HandJ} is aided by the following
lemma 
\begin{lem}
\label{lem:EFrechetIntegral}Let $X$ be $1$-Fr\'{e}chet
with scale $v_{0}$, i.e. $\mathbb{P}\left(X\le x\right)=e^{-v_{0}/x},x>0.$
Then 

\begin{enumerate}

\item[(i)] 
\begin{equation}
\mathbb{E}\sqrt{X}=\sqrt{\pi v_{0}}\label{eq:EsqrtFrechet}
\end{equation}

\item[(ii)]

\begin{equation}
\mathbb{E}\gamma_{\frac{1}{2}}\left(v/X\right)=\sqrt{\frac{\pi\left(v_{0}+v\right)}{v}}\label{eq:Egamma-half}
\end{equation}
 so that in particular $\mathbb{E}\gamma_{\frac{1}{2}}\left(v_{0}/X\right)=\sqrt{\pi/2}$.

\item[(iii)]
\begin{eqnarray}
\mathbb{E}\mathfrak{F}\left(X,v\right) & = & 2\sqrt{\pi}\left(2\sqrt{v_{0}+v}-\sqrt{v_{0}}-\sqrt{2v}\right)\label{eq:EFrechetIntegral}
\end{eqnarray}

\end{enumerate}\end{lem}
\begin{proof}
For \eqref{eq:EsqrtFrechet} note that $\sqrt{X}$ is equal in distribution
to a $2$-Fr\'{e}chet random variable with scale $v_{0}$ which has
finite expectation $\sqrt{\pi v_{0}}.$ For \eqref{eq:Egamma-half},
applying Fubini's Theorem, and observing that $\mathbb{E}\left(\mathbf{1}_{\left\{ X\le v/s\right\} }\right)=e^{-v_{0}s/v}$,
with $\Gamma(1/2) = \sqrt{\pi}$
\begin{align*}
\mathbb{E}\gamma_{\frac{1}{2}}\left(v/X\right) & =\int_{0}^{\infty}e^{-v_{0}s/v}s^{1/2-1}e^{-s}ds\\
 & =\int_{0}^{\infty}s^{1/2-1}e^{-vs/(v_{0}+v)}ds = \sqrt{\frac{\pi v}{v_{0}+v}}.
\end{align*}
This establishes \eqref{eq:Egamma-half}. For \eqref{eq:EFrechetIntegral},
substituting the espression $\mathfrak{F}\left(X,v\right)$ from Lemma
\ref{lem:FrechetIntegral} we have
\[
\mathbb{E}\mathfrak{F}\left(X,v\right)=4\mathbb{E}\left[\sqrt{X}\left(e^{-v/X}-\frac{1}{2}\right)+\sqrt{v}\left(\gamma_{\frac{1}{2}}\left(v/X\right)-\sqrt{\frac{\pi}{2}}\right)\right]
\]
which, after substituting \eqref{eq:EsqrtFrechet} and \eqref{eq:Egamma-half}
yeilds 
\begin{equation}
\mathbb{E}\mathfrak{F}\left(X,v\right)=4\left[\mathbb{E}\sqrt{X}e^{-v/X}-\frac{\sqrt{\pi v_{0}}}{2}+v\sqrt{\frac{\pi}{v_{0}+v}}+-\sqrt{\frac{\pi v}{2}}\right].\label{eq:Efrechetint1}
\end{equation}
 Using the fact that $X$ is distributed $1$-Fr\'{e}chet with scale
$v_{0}$ we have
\begin{align*}
\mathbb{E}\sqrt{X}e^{-v/X} & =\int_{0}^{\infty}\sqrt{s}e^{-v/s}v_{0}e^{-v_{0}/s}s^{-2}ds\\
 & =v_{0}\int_{0}^{\infty}s^{-3/2}e^{-\left(v_{0}+v\right)/s}ds.
\end{align*}
Now the substitution $t=s^{-1}$ gives 
\begin{equation}
\mathbb{E}\sqrt{X}e^{-v/X}=v_{0}\int_{0}^{\infty}t^{-1/2}e^{-\left(v_{0}+v\right)t}dt=\frac{\sqrt{\pi}v_{0}}{\sqrt{v_{0}+v}}.\label{eq:EfrechetInt2}
\end{equation}
Plugging \eqref{eq:EfrechetInt2} into \eqref{eq:Efrechetint1} yields
\begin{align*}
\mathbb{E}\mathfrak{F}\left(X,v\right) & =4\left[\frac{\sqrt{\pi}v_{0}}{\sqrt{v_{0}+v}}-\frac{\sqrt{\pi v_{0}}}{2}+v\sqrt{\frac{\pi}{v_{0}+v}}-\sqrt{\frac{\pi v}{2}}\right]\\
 & =2\sqrt{\pi}\left[\frac{2v_{0}}{\sqrt{v_{0}+v}}-\sqrt{v_{0}}+\frac{2v}{\sqrt{v_{0}+v}}-\sqrt{2v}\right]\\
 & =\eqref{eq:EFrechetIntegral}.
\end{align*}

\end{proof}

\begin{proof}[Proof of Corollary \ref{cor:HandJ}]
\label{sub:Numerical-integration-details-appendix} 

Recall $M_{u}:=\max_{t\in D}\left\{ X_{t}/u_{t}\right\} $ and 
\begin{equation}
H_{\theta_{0}}=\left.\frac{\partial^{2}}{\partial\theta\partial\theta^{\top}}\mathbb{E}\mathcal{E}_{\theta}\left(X\right)\right|_{\theta=\theta_{0}}.\label{eq:discreteAsymBread1}
\end{equation}
Substituting \eqref{eq:maxstablecrps3} gives
\[
H_{\theta_{0}} =\left.\sum_{u\in\mathcal{U}}\frac{\partial^{2}}{\partial\theta\partial\theta^{\top}}\mathbb{E}\mathfrak{F}\left(M_{u},V_{\theta}\left(u\right)\right)\right|_{\theta=\theta_{0}}.
\]
Note that Lemma \ref{lem:EFrechetIntegral} implies
\[
\mathbb{E}\mathfrak{F}\left(M_{u},V_{\theta}\left(u\right)\right)=2\sqrt{\pi}\left(2\sqrt{V_{\theta_{0}}\left(u\right)+V_{\theta}\left(u\right)}-\sqrt{V_{\theta_{0}}\left(u\right)}-\sqrt{2V_{\theta}\left(u\right)}\right)
\]
from which it follows
\[
\left.\frac{\partial^{2}}{\partial\theta\partial\theta^{\top}}\mathbb{E}\mathfrak{F}\left(M_{u},V_{\theta}\left(u\right)\right)\right|_{\theta=\theta_{0}}=\frac{\sqrt{\pi}}{2}\left(2V_{\theta_{0}}\left(u\right)\right)^{-3/2}\dot{V}_{\theta_{0}}\left(u\right)\left(\dot{V}_{\theta_{0}}\left(u\right)\right)^{\top}
\]
which completes the proof of \eqref{eq:Cor1bread}.

Now recall that $J_{\theta_{0}}=\mathbb{E}\left\{ \dot{\mathcal{E}}_{\theta_{0}}\left(X\right)\dot{\mathcal{E}}_{\theta_{0}}\left(X\right)^{\top}\right\}$
Substituting \eqref{eq:maxstablecrps3}, we obtain
\begin{align}
J_{\theta_{0}} & =\mathbb{E}\sum_{u,w\in\mathcal{U}}\dot{\mathfrak{F}}\left(M_{u},V_{\theta_0}\left(u\right)\right)\dot{\mathfrak{F}}\left(M_{w},V_{\theta_0}\left(w\right)\right)\dot{V}_{\theta_{0}}\left(u\right)\left(\dot{V}_{\theta_{0}}\left(w\right)\right)^{\top}\\
 & =\sum_{u,w\in\mathcal{U}}\mathbb{E}\left\{ \dot{\mathfrak{F}}\left(M_{u},V_{\theta_0}\left(u\right)\right)\dot{\mathfrak{F}}\left(M_{w},V_{\theta_0}\left(w\right)\right)\right\} \dot{V}_{\theta_{0}}\left(u\right)\left(\dot{V}_{\theta_{0}}\left(w\right)\right)^{\top}\label{eq:proofCor1meat4}
\end{align}
where, in view of Lemma \ref{lem:FrechetIntegral}, on can show that 
\begin{equation}
\mathfrak{\dot{F}}\left(M_{u},V_{\theta}\left(u\right)\right) \equiv \partial_v \mathfrak{F}\left(M_{u},V_{\theta}\left(u\right)\right)=\frac{\sqrt{\pi/2}-\gamma_{\frac{1}{2}}\left(V_{\theta}\left(u\right)/M_{u}\right)}{\sqrt{V_{\theta}\left(u\right)}}\label{eq:proofCor1meat3}.
\end{equation}
 Our next goal is to calculate 
\begin{equation}
\mathbb{E}\left\{ \mathfrak{\dot{F}}\left(M_{u},V_{\theta}\left(u\right)\right)\mathfrak{\dot{F}}\left(M_{w},V_{\theta}\left(w\right)\right)\right\} \label{eq:discreteAsymMeat2}
\end{equation}
where the expectation is taken under $\theta_{0}$. Using the fact
that $M_{u}$ is 1-Fr\'{e}chet with scale $V_{\theta_{0}}\left(u\right),$
Lemma \ref{lem:EFrechetIntegral}\emph{(ii) }implies 
\begin{equation}
\mathbb{E}\gamma_{\frac{1}{2}}\left(V_{\theta_0}\left(u\right)/M_{u}\right)=\frac{\sqrt{\pi V_{\theta_0}\left(u\right)}}{\sqrt{V_{\theta_{0}}\left(u\right)+V_{\theta}\left(u\right)}}=\sqrt{\frac{\pi}{2}}.\label{eq:discreteMeat2.1}
\end{equation}
Thus, in view of \eqref{eq:proofCor1meat3},  \eqref{eq:discreteAsymMeat2} becomes
\begin{multline*}
\frac{\mathbb{E}\left\{ \gamma_{\frac{1}{2}}\left(V_{\theta}\left(u\right)/M_{u}\right)\gamma_{\frac{1}{2}}\left(V_{\theta}\left(w\right)/M_{w}\right)\right\} |_{\theta=\theta_{0}}-\frac{\pi}{2}}{\sqrt{V_{\theta}\left(u\right)}\sqrt{V_{\theta}\left(w\right)}} = \\ \frac{\mathrm{Cov}\left\{ \gamma_{\frac{1}{2}}\left(V_{\theta}\left(u\right)/M_{u}\right),\gamma_{\frac{1}{2}}\left(V_{\theta}\left(w\right)/M_{w}\right)\right\}}{\sqrt{V_{\theta}\left(u\right)}\sqrt{V_{\theta}\left(w\right)}}.
\end{multline*}
This, in view of \eqref{eq:proofCor1meat4} implies \eqref{eq:Cor1meat} and completes the proof.
\end{proof}

\bibliographystyle{spbasic}

\end{document}